\documentclass[draftclsnofoot,onecolumn, 12pt]{IEEEtran}
\usepackage{color}
\usepackage{graphicx}
\usepackage{amsmath}
\usepackage{amsfonts}
\usepackage{amssymb}
\usepackage{booktabs}
\usepackage{stfloats}
\usepackage{cases}
\usepackage{tabularx}
\usepackage{boxedminipage}
\usepackage{url}

\newcommand{\beq}{ \begin{equation} }
\newcommand{\eeq}{ \end{equation} }
\newcommand{\beqa}{\begin{eqnarray}}%
\newcommand{\eeqa}{\end{eqnarray}}

\graphicspath{{./files/}}

\hoffset- 8mm
\begin{document}

\title{Buffer-Aided Relaying with Adaptive Link Selection}
\author{Nikola Zlatanov,   Robert Schober, and Petar Popovski 
\thanks{This paper has been presented in part at IEEE Globecom 2011, Houston, December 2011.}
\thanks{N. Zlatanov and R. Schober are with the Department of Electrical and Computer Engineering, University of British Columbia, Vancouver, BC V6T 1Z4,
Canada, E-mail: zlatanov@ece.ubc.ca, rschober@ece.ubc.ca}
\thanks{Petar Popovski is with the Department of Electronic Systems, Aalborg University, Denmark, E-mail:
petarp@es.aau.dk.}
}

\maketitle
\vspace{-1.5cm}
\begin{abstract}
In this paper, we consider a simple network consisting of a source, a half-duplex decode-and-forward relay, and a destination. We propose a new relaying protocol employing adaptive link selection, i.e., in any given time slot, based on the channel state information of the source-relay and the relay-destination link a decision is made whether the source or the relay transmits. In order to avoid data loss at the relay, adaptive link selection requires the relay to be equipped with a buffer such that data can be queued 
until the relay-destination link is selected for transmission. We study both delay constrained and delay unconstrained transmission. For the delay unconstrained case, we characterize the optimal link selection policy, derive the corresponding throughput, and develop an optimal 
power allocation scheme. For the delay constrained case, we propose to starve the buffer of the relay by choosing the decision threshold of the link selection policy smaller than the optimal one and derive a corresponding upper bound on the average delay. Furthermore, we propose a 
modified link selection protocol which avoids buffer overflow by limiting the queue size.  Our analytical and numerical results show that buffer-aided relaying with adaptive link selection achieves significant throughput gains compared to 
conventional relaying protocols with and without buffers where the relay employs a fixed schedule for reception and transmission.
\end{abstract}



\newtheorem{theorem}{Theorem}
\newtheorem{corollary}{Corollary}
\newtheorem{remark}{Remark}
\newtheorem{lemma}{Lemma}
\section{Introduction}
The classical three-node relay channel was originally considered by van der Meulen \cite{meulen}. Cover and El Gamal  \cite{cover} investigated the capacity of a memoryless relay  channel consisting of a source, a
destination, and a single full-duplex relay. Recently, renewed interest in relay-assisted communication was sparked by \cite{erkip1} and \cite{erkip2}. Since then, the simple three-node system has become a building 
block for more sophisticated relay networks and a host of cooperative communication techniques have been proposed \nocite{1435648,laneman1,laneman2,nostrina,1542409,1611050,1499041,eth} \cite{1435648}-\cite{eth}.  
The capacity of a three-node relay channel comprised of a source, a destination, and a single half-duplex decode-and-forward (DF) relay was investigated in \cite{1435648}. Under the assumption of full channel state information 
(CSI) at every node, it was shown in \cite{1435648} that the capacity of a three-node DF network without a direct source-destination link is given by the minimum of the source-relay  and the relay-destination link capacities for the case 
when the relay employs a fixed schedule for reception and transmission.
In \cite{1435648}-\cite{eth} and most of the existing literature on half-duplex relaying, it is assumed that the relays receive a packet from the source in one time slot and forward it to the destination in the next time slot. Relay protocols operating
under this restriction are referred to as ``conventional" relay protocols in the following. In this paper, we abandon this paradigm and give relays the freedom to decide in which time slot to receive and in which time slot to transmit. This new 
approach requires the relays to have buffers. Relays with buffers were also considered in \cite{XFTP08} and \cite{Aissa11_J}. In \cite{XFTP08}, the buffer at the relay was used to enable the relay to receive for a fixed number of time slots before 
retransmitting the received information. In \cite{Aissa11_J}, relay selection was considered and buffers enable the selection of the relay with the best source-relay channel for reception and the relay with the best relay-destination 
channel for transmission. However, similar to conventional relay selection without buffers \cite{Bletsas06}, the source transmits in every other time slot. Thus, both \cite{XFTP08} and \cite{Aissa11_J} do not fully exploit the flexibility offered by relays 
with buffers since the schedule of when the source transmits and when a relay transmits is a priori fixed. 

In this paper, we consider buffer-aided relaying with adaptive link selection, where in any given time slot based on the CSI of the source-relay and the relay-destination link a decision is made whether the source or the relay transmits.  
We consider the cases of delay constrained and delay unconstrained transmission. For the delay unconstrained case, we optimize the link selection protocol and the power allocated to the source and the relay. 
Interestingly, the optimal link selection policy requires only knowledge of the instantaneous CSI of the considered time slot and the statistical CSI of the involved links. However, the instantaneous CSI of past and future time slots and the state of the relay's 
buffer are not required for optimal link selection, which facilitates the implementation of the optimal policy. For the delay constrained case, we propose two alternative link selection protocols and provide an upper bound for the average delay of one of them. Our analytical 
and simulation results show, in good agreement, that buffer-aided relaying with adaptive link selection can achieve large performance gains compared to conventional relaying with or without buffer \cite{XFTP08}, as long  as a certain delay can be tolerated.

The remainder of the paper is organized as follows. In Section \ref{sys-mod}, the considered system and channel models are presented. The proposed link selection protocol for buffer-aided relaying is introduced and optimized in Section \ref{var-r-sol}, and
optimal power allocation for source and relay is discussed in Section  \ref{s34}.  In Section \ref{sec-delay}, we propose two protocols for buffer-aided relaying with delay constraints. Numerical results are presented in Section \ref{numerics}, and some 
conclusions are drawn in Section \ref{conclude}.
\section{System Model}\label{sys-mod}
We consider a three-node communication system comprising a source $\mathcal{S}$, a half-duplex relay $\mathcal{R}$, and a destination $\mathcal{D}$, cf.~Fig.~\ref{fig1}. The source can communicate with the destination only through the relay, i.e., 
there is no direct $\mathcal{S}$-$\mathcal{D}$ link. The source sends packets to the relay, which decodes these packets, possibly stores them in its buffer, and eventually sends them to the destination. Throughout this paper, we assume that the source 
has always data to transmit. 
\subsection{Channel Model}\label{chan-mod}
We assume that time is divided into slots of equal lengths. In the $i$th time slot, the transmit powers of source and relay are denoted by $P_S(i)$ and  $P_R(i)$, respectively, and the instantaneous (squared) channel gains of the 
$\mathcal{S}$-$\mathcal{R}$ and $\mathcal{R}$-$\mathcal{D}$ links are denoted by $h_S(i)$ and $h_R(i)$, respectively. $h_S(i)$ and $h_R(i)$ are modeled as mutually independent, non-negative, stationary, and ergodic random processes 
with expected values $E\{h_S(i)\}\triangleq \bar\Omega_S$ and $E\{h_R(i)\}\triangleq \bar\Omega_R$, where $E\{\cdot\}$ denotes expectation. We assume that the channel gains are constant 
during one time slot but change from one time slot to the next due to e.g.~the mobility of the involved nodes and/or frequency hopping. 

The instantaneous link signal-to-noise ratios (SNRs) of the $\mathcal{S}$-$\mathcal{R}$ 
and $\mathcal{R}$-$\mathcal{D}$ channels in the $i$th time slot are given by $s(i)\triangleq \gamma_S(i)  h_S(i)$ and $r(i)\triangleq \gamma_R(i)  h_R(i)$, respectively. 
Here, $\gamma_S(i)=P_S(i)/\sigma_{n_S}^2$ and $\gamma_R(i)=P_R(i)/\sigma_{n_R}^2$ denote the transmit SNRs of the source and the relay, respectively, and $\sigma_{n_S}^2$ and $\sigma_{n_R}^2$ are 
the variances of the additive white Gaussian noise (AWGN) at the relay and the destination, respectively. The average link SNRs are denoted by $\Omega_S\triangleq E\{s(i)\}$ and  $\Omega_R\triangleq E\{r(i)\}$.
 
\subsection{Link Adaptive Transmission Protocol}
In the proposed link adaptive transmission protocol, one of the nodes of the network is responsible for deciding whether the source or the relay should transmit in a given time slot $i$. This node is referred to as the central node in the following. 
The central node broadcasts its decision to the other nodes before transmission in time slot $i$ begins. If they are selected for transmission, the source and the relay adapt their transmission rates to the capacity of the respective link and transmit codewords
spanning one time slot. We assume that source and relay employ capacity-achieving codes. For both link selection and rate adaptation, the nodes require CSI knowledge as will be detailed in the following.

\textbf{CSI requirements:} The central node requires knowledge of the instantaneous channel gains $h_S(i)$ and $h_R(i)$. In addition, regardless of which node is the central node, if the $\mathcal{S}$-$\mathcal{R}$ link is selected, 
the source and the relay require knowledge of $h_S(i)$ for rate adaptation and decoding, respectively. On the other hand, if the $\mathcal{R}$-$\mathcal{D}$ link is selected, the relay and the destination require knowledge of $h_R(i)$ for rate 
adaptation and decoding, respectively. Nodes can obtain the instantaneous channel gains through estimation based on pilot symbols emitted by other nodes and/or CSI feedback from other nodes. Furthermore, we assume that the central
node knows the noise variances $\sigma_{n_S}^2$ and $\sigma_{n_R}^2$, and regardless of which node is the central node, source, relay, and destination know  $\sigma_{n_S}^2$, $(\sigma_{n_S}^2, \sigma_{n_R}^2)$, and $\sigma_{n_R}^2$, 
respectively. Also, if the transmit power is a priori fixed, i.e., $P_S(i)=P_S$ and $P_R(i)=P_R$, $\forall i$, the central node requires knowledge of $P_S$ and $P_R$, and source, relay, and destination require knowledge of $P_S$, $(P_S,P_R)$, 
and $P_R$, respectively. If power allocation is employed, $P_S(i)$ and $P_R(i)$ are computed by the nodes based on their respective knowledge of the instantaneous channel gains and statistical CSI knowledge (cf.~Section \ref{s34}, (\ref{power-eq-2a}),
(\ref{power-eq-2b})).

Which node serves as the central node depends on the network architecture. For example, in the downlink of a cellular network, the source (base station) can serve as the central node as it typically acquires the full
CSI of all links anyways and can afford the complexity of performing adaptive link selection and power allocation. On the other hand, if the relay serves as the central node, source and destination have to acquire only the
CSI of the $\cal S$-$\cal R$ and $\cal R$-$\cal D$ links, respectively.

For convenience we normalize the number of bits transmitted in one time slot by the number of symbols per time slot. Thus, throughout the remainder of this paper, the number of bits refers to the number of bits divided by the
number of symbols in a codeword. In the following, we discuss the network dynamics when source and relay transmit, respectively.

\textbf{Source transmits:} If the source is selected for transmission in time slot $i$, it transmits with rate 
\begin{eqnarray}\label{rav1}
    S_{\mathcal{S}\mathcal{R}}(i)=\log_2(1+ s(i)).
\end{eqnarray}
Hence, the relay 
receives $S_{\mathcal{S}\mathcal{R}}(i)$ data bits from the source and appends them to the queue in its buffer. The (normalized) number of bits in the buffer of the relay at the end of time slot $i$ is denoted by $Q(i)$ and given by 
\begin{eqnarray}\label{eq2}
    Q(i)=Q(i-1)+ S_{\mathcal{S}\mathcal{R}}(i).
\end{eqnarray}

\textbf{Relay transmits:} If the relay transmits in time slot $i$, the number of bits transmitted by  the relay is given by
\begin{eqnarray}\label{3-lit}
    R_{\mathcal{R}\mathcal{D}}(i)=\min\{\log_2(1+r(i)),Q(i-1)\},
\end{eqnarray}
where we take into account that the maximal number of bits that can be send by the relay is limited by the number of bits in its buffer and the instantaneous capacity of the $\mathcal{R}$-$\mathcal{D}$ link. The number of bits remaining in the buffer at the end of time slot $i$ is given by
\begin{eqnarray}\label{eq4}
    Q(i)=Q(i-1)-R_{\mathcal{R}\mathcal{D}}(i),
\end{eqnarray}
which is always non-negative because of (\ref{3-lit}). 

Because of the half-duplex constraint, we have $R_{\mathcal{R}\mathcal{D}}(i) = 0$ when the source transmits (and the relay listens), and we have $S_{\mathcal{S}\mathcal{R}}(i)= 0$ when the relay transmits.
\subsection{Throughput}
Since we assume the source has always data to transmit, the average (normalized) number of bits that arrive at the destination per time slot is given by 
\begin{eqnarray}\label{eq1}
    \tau=\lim_{N\to\infty}\frac{1}{N}\sum_{i=1}^N R_{\mathcal{R}\mathcal{D}}(i),
\end{eqnarray}
i.e., $\tau$ is the throughput of the considered communication system. The goal of the following sections is the maximization of $\tau$ by optimizing the adaptive link selection protocol and the transmit power allocated to source and relay.
\subsection{Conventional Relaying}
For comparison purpose, we provide the throughput of two baseline schemes. Thereby, we assume that the transmit powers at the source and the relay are fixed, i.e., $P_S(i)=P_S$, $P_R(i)=P_R$, $\forall i$.
\subsubsection{Conventional relaying without buffer} 
The instantaneous throughput of conventional relaying without buffer, where the relay receives a packet in one time slot and transmits it in the next, is given in \cite{1435648}, and the corresponding average throughput is 
  \begin{eqnarray}\label{tau-conv1}
    \tau_{\rm conv,1}&\triangleq&\lim_{N\to\infty}\frac{1}{N}\sum_{i=1}^{N/2}\min\{\log_2(1+s(2i-1)),\log_2(1+r(2i))\}\nonumber\\
&=&\frac{1}{2}E\{\min\{\log_2(1+s(i)),\log_2(1+r(i))\}\},
\end{eqnarray}
where the ergodicity of $s(i)$ and $r(i)$ was exploited. Note that conventional relaying without buffer introduces a delay of one time slot since the relay has to wait until the entire codeword is received and decoded before sending the codeword to the destination.

\subsubsection{Conventional relaying with buffer \cite{XFTP08}} 
In conventional relaying with buffer as proposed in \cite{XFTP08}, the relay receives data from the source in the first $N/2$ ($N$ is even) time slots and sends this cumulative information to the destination in the next $N/2$ slots.  
The corresponding maximum achievable average throughput is obtained for $N\to \infty$ and given by
  \begin{eqnarray}\label{tau-conv2}
    \tau_{\rm conv,2} &\triangleq& \lim_{N\to\infty} \frac{1}{N} \min\left\{\sum_{i=1}^{N/2}\log_2(1+s(i)),\sum_{i=N/2}^{N}\log_2(1+r(i))\right\} \nonumber\\
&=& \frac 12\min\{E\{\log_2(1+s(i))\},E\{\log_2(1+r(i))\}\}.
\end{eqnarray} 
Comparing (\ref{tau-conv1}) and (\ref{tau-conv2}), we observe that $\tau_{\rm conv,2}\geq \tau_{\rm conv,1}$ holds \cite{XFTP08}. However, to realize this performance gain, the relay has to be equipped with a  buffer of infinite size and an infinite delay is introduced.

\subsubsection{Rayleigh fading} For the numerical results shown in Section \ref{numerics}, we consider the case where the $\mathcal{S}$-$\mathcal{R}$ and $\mathcal{R}$-$\mathcal{D}$ links are both Rayleigh faded, i.e., the probability density functions (pdfs) 
of $s(i)$ and $r(i)$ are given by $f_s(s)=e^{-s/\Omega_S}/\Omega_S$ and $f_r(r)=e^{-r/\Omega_R}/\Omega_R$, respectively. In this case, $\tau_{\rm conv, 1}$ and $\tau_{\rm conv, 2}$ can be obtained in closed form as
\begin{eqnarray}\label{t-c-1}
    \tau_{\rm conv,1}=\frac{1}{2\ln(2)}  
\exp\left(\frac{\Omega_R+  \Omega_S}{  \Omega_S \Omega_R} \right)
E_1 \left(\frac{\Omega_R+  \Omega_S}{  \Omega_S \Omega_R} \right)
\end{eqnarray}
and
\begin{eqnarray}\label{th-conv}
    \tau_{\rm conv,2}=\frac{1}{2\ln(2)} \min\left\{ \exp\left(\frac{1}{  \Omega_S } \right) E_1\left(\frac{1}{  \Omega_S } \right) ,
 \exp\left(\frac{1}{  \Omega_R } \right) E_1\left(\frac{1}{  \Omega_R } \right) \right\},
\end{eqnarray}
respectively, where $E_1(x)=\int_x^\infty e^{-t}/t\,dt$, $x>0$, denotes the exponential integral function. 
\section{Adaptive Link Selection}\label{var-r-sol}
To gain some insight, we assume throughout this section constant source and relay powers, i.e., $P_S(i)=P_S$, $P_R(i)=P_R$, $\forall i$, and a buffer of unlimited size at the relay. 
For this case, we derive the optimal link selection policy and the corresponding throughput. Optimal power allocation and the effect of a limited buffer size 
will be discussed in Sections \ref{s34} and \ref{sec-delay}, respectively.
\subsection{Problem Formulation}
Let $d_i\in\{0,1\}$ denote a binary decision variable. We set $d_i=1$ if the $\mathcal{R}$-$\mathcal{D}$ channel is selected for transmission in time slot $i$, i.e., the relay transmits and the destination receives. Similarly, we set $d_i=0$ if
the $\mathcal{S}$-$\mathcal{R}$ channel is selected for transmission in time slot $i$, i.e., the source transmits and the relay receives. Exploiting $d_i$, the number of bits send from the source to the relay and from the relay to the destination
in time slot $i$ can be written in compact form as
\begin{eqnarray}\label{s_2}
    S_{\mathcal{S}\mathcal{R}}(i)=(1-d_i)S(i)
\end{eqnarray}
and
\begin{eqnarray}\label{s_1}
    R_{\mathcal{R}\mathcal{D}}(i)=d_i\min\{R(i),Q(i-1)\}, 
\end{eqnarray}
respectively, where $S(i)\triangleq \log_2(1+s(i))$ and $R(i)\triangleq \log_2(1+ r (i))$. Consequently, the throughput in (\ref{eq1}) can be rewritten as
\begin{eqnarray}\label{trup}
    \tau=\lim_{N\to\infty}\frac{1}{N}\sum_{i=1}^N d_i \min\{R(i),Q(i-1)\}.
\end{eqnarray}
The considered problem can now be stated as follows: Find the optimal link selection policy, i.e., the optimal sequence $d_i$, $i\ge 1$, which maximizes throughput $\tau$.
\begin{remark}\label{remark0}
Our problem formulation is quite general in the sense that we have introduced no restrictions concerning the required knowledge regarding the channel and the queue states. In other words, the optimal decision at time $j$, $d_j$, potentially
depends on $h_S(i)$, $h_R(i)$, and $Q(i)$, $i\ge1$, and thus requires non-causal channel knowledge. Fortunately, as will be shown in the ensuing section, this is not the case and the optimal policy turns out to be rather simple and easy to implement.
\end{remark}
\subsection{Optimal Link Selection Policy \label{s32}}
Let us first define the average arrival rate of bits per slot into the queue of the buffer, $A$, and the average departure rate of bits per slot out of the queue of the buffer, $D$, as \cite{2}
\begin{eqnarray}
    A\triangleq\lim_{N\to\infty}\frac{1}{N}\sum_{i=1}^N (1-d_i) S(i) \label{AD}
\end{eqnarray}
and
\begin{eqnarray}
 D\triangleq\lim_{N\to\infty}\frac{1}{N}\sum_{i=1}^N  d_i \min\{R(i),Q(i-1)\},\label{D}
\end{eqnarray}
respectively. We note that the average departure rate of the queue is equal to the throughput. The queue is said to be absorbing if $A>D=\tau$. The following theorem characterizes the optimal link selection policy in terms of the state of the queue in the buffer of the relay.

\begin{theorem}\label{theorem1}
A necessary condition for the optimal link selection policy, which maximizes the throughput, is that the queue in the buffer of the relay is at the edge of non-absorbtion, i.e., the queue is non-absorbing but is at the boundary of a non-absorbing
and an absorbing queue.  
\end{theorem}
\begin{proof}
Please refer to Appendix \ref{app_A}.
\end{proof}
Exploiting Theorem \ref{theorem1}, we can establish a useful condition that the optimal link selection policy has to fulfill and a simplified expression for the throughput. This is the subject of the following theorem. 
\begin{theorem}\label{theorem2}
Assuming non-negative, ergodic, and stationary random processes $s(i)$ and $r(i)$, for the  optimal link selection policy the identity
\begin{eqnarray}\label{max-solution}
E\{(1-d_i) S(i)\}=E\{d_i R(i)\}
\end{eqnarray}
holds and the throughput is then given by
\beq\label{max-tau}
\tau = E\{d_i R(i)\}.
\eeq
\end{theorem}
\begin{proof}
Please refer to Appendix \ref{app_B}.
\end{proof}

\begin{remark}\label{remark1}
According to Theorem \ref{theorem2}, for the optimal link selection policy, the queue is non-absorbing but is almost always filled to such a level that the number of bits in the queue exceed the number of bits that can be transmitted 
over the $\mathcal{R}$-$\mathcal{D}$ channel. In particular, as shown in Appendix \ref{app_B}, condition (\ref{max-solution}) automatically ensures that for $N\to\infty$, $\frac{1}{N} \sum_{i=1}^Nd_iR(i)= \frac{1}{N} \sum_{i=1}^Nd_i\min\{R(i),Q(i-1)\}$
is valid, i.e., the impact of event $R(i)>Q(i-1)$, $i=1,\ldots,N$, is negligible.
\end{remark}

We are now ready to derive the optimal link selection policy for buffer-aided relaying. According to Theorem \ref{theorem2},  the policy that maximizes the
throughput $\tau$ in (\ref{max-tau}) can only be found inside the set of policies that produce a queue which satisfies (\ref{max-solution}), and not
outside of this set of policies. 
 Thus, for $N\to\infty$, we formulate the following optimization problem:
\begin{eqnarray}\label{MPR1}
\begin{array}{ll}
 {\underset{d_i}{\rm{Maximize: }}}&\frac{1}{N}\sum_{i=1}^N d_i R(i)\\
{\rm{Subject\;\; to: }} &{\rm C1:}\, \frac{1}{N}\sum_{i=1}^N (1-d_i) S(i)=\frac{1}{N}\sum_{i=1}^N d_i R(i)\\
  &{\rm C2:} \frac{1}{N}d_i(1-d_i)=0\quad \forall i\\
\end{array}
\end{eqnarray}
where constraint C1 ensures that the search for the optimal policy is conducted only among those policies that satisfy (\ref{max-solution}) and C2 ensures that $d_i\in\{0,1\}$. We note that C1 and C2 do not exclude the case that the relay 
is chosen for transmission if $R(i)>Q(i-1)$. However, according to Remark \ref{remark1}, C1 ensures that the influence of event $R(i)>Q(i-1)$ is negligible. Therefore, an additional constraint dealing with this event is not required.
The solution of problem (\ref{MPR1}) leads to the following theorem.
\begin{theorem}\label{theorem3}
The optimal policy maximizing the throughput of buffer-aided relaying with adaptive link selection is given by
\begin{eqnarray}\label{kraj}
d_i=\left\{
\begin{array}{cc}
1 & \textrm{if }   \mathcal{F}(r(i))\geq \rho \mathcal{F}(s(i))\\
0 & \textrm{otherwise}
\end{array} 
\right.
\end{eqnarray} 
where $\rho$ is referred to as the decision threshold and the optimal decision function is given by
\beq
\mathcal{F}(x) = \log_2(1+x). 
\label{xx}
\eeq
The optimal decision threshold $\rho$, $\rho_{\rm opt}$, has to satisfy (\ref{max-solution}).
\end{theorem}
\begin{proof}
Please refer to Appendix \ref{app_C}.
\end{proof}
\begin{remark}\label{remark3}
Interestingly, we observe from Theorem \ref{theorem3} that the optimal decision, $d_i$, at time slot $i$, depends only on the instantaneous SNRs, $s(i)$ and $r(i)$, of that time slot. $d_i$ does not depend on the state of the queue, $Q(i)$, in any time slot nor on the instantaneous SNRs in previous or future time slots. This makes the proposed optimal selection policy highly practical. We note that the optimal decision threshold, $\rho_{\rm opt}$, depends on the statistical CSI of both involved links as will be established in the next section. The independence of the optimal link selection policy from non-causal instantaneous CSI is caused by the relay being operated at the edge of non-absorbtion, i.e., the relay node is practically fully backlogged. Non-causal knowledge would only help 
buffer management (i.e., ensuring that there is a sufficient number of bits in the buffer for upcoming time slots), which is not required in the considered regime. 
\end{remark}
\begin{remark}\label{remark3a}
In this paper, we assume that the transmitting nodes have perfect CSI and apply adaptive rate transmission. However, we note that this is not necessary for achieving the maximum throughput in (\ref{max-tau}). In fact, the proposed adaptive link selection protocol (\ref{kraj})
also achieves the maximum  throughput in (\ref{max-tau}) if source and relay transmit long codewords that span (ideally infinitely) many time slots (and consequently many channel states). In this case, both the source and the relay can transmit with constant rate $\tau=E\{(1-d_i)S(i)\}=E\{d_iR(i)\}$ and rate adaptation is not necessary. The first codeword is transmitted by the source without link adaptation and decoded by the relay. For all subsequent codewords, link adaptation is performed based on (\ref{kraj})
and source and relay transmit parts of a long codeword 
whenever they are selected for transmission. The disadvantage of this approach is that the long codewords inherently introduce (ideally infinitely) long delays and the generalization of this approach to the delay constrained case is difficult. Therefore, in this paper, 
we consider adaptive rate transmission and assume that one codeword spans only one time slot (and consequently one channel state).
\end{remark}
%
\subsection{Generalization of the Decision Function and Optimal Decision Threshold\label{s33}}
The optimal decision function in (\ref{xx}) is intuitively pleasing since it basically means that link selection is performed based on the instantaneous link capacities. 
Nevertheless, for some fading distributions, the optimal decision function may lead to complicated expressions for the throughput and the optimal decision threshold $\rho_{\rm opt}$. In such cases,
simpler suboptimal decision functions such as $\mathcal{F}(x) = x$ may be preferable as they generally lead to a similar performance as $\mathcal{F}(x) = \log_2(1+x)$ but
are analytically more tractable. Thus, in the following, we generalize the decision function, $\mathcal{F}(x)$, to be any non-negative, smooth, and increasing function, i.e., 
$\mathcal{F}(x+\epsilon)>\mathcal{F}(x)$ for $\epsilon>0$. We assume that the inverse of $\mathcal{F}(\cdot)$ exists and denote it by $\mathcal{F}^{-1} (\cdot)$.

We note that Theorems \ref{theorem1} and \ref{theorem2} also hold for suboptimal decision functions $\mathcal{F}(x)$. For a given $\mathcal{F}(x)$, the corresponding decision
threshold $\rho$ has to satisfy  (\ref{max-solution}). Thus, the optimal decision threshold, $\rho_{\rm opt}$, for a given (optimal or
suboptimal) decision function, $\mathcal{F}(x)$, can be computed based on the following lemma.

\begin{lemma}\label{lemma1}
Denote the pdfs of $s(i)$ and $r(i)$ by $f_s(s)$ and $f_r(r)$, respectively. For a given decision function, $\mathcal{F}(x)$,
the optimal $\rho_{\rm opt}$ is the solution of 
\begin{eqnarray}\label{resenija}
\int_{0}^\infty \left[\int_{G(r)}^\infty  \log_2(1+s) f_s(s) ds\right] f_r(r) dr= \int_0^{\infty } \left[\int_{H(s)}^{\infty }  \log_2(1+r) f_r(r) dr\right] f_s(s) ds,
\end{eqnarray}
where the integral limits are given by $G(r)\triangleq\mathcal{F}^{-1}(\mathcal{F}(r)/\rho)$ and $H(s)\triangleq \mathcal{F}^{-1}(\rho\mathcal{F}(s))$.
\end{lemma}
\begin{proof}
The left hand side of (\ref{max-solution}) is the expectation of variable $(1-d_i)\log_2(1+s(i))$. This variable is nonzero only when  $d_i=0$. From (\ref{kraj}) we observe that $d_i=0$  if $\rho \mathcal{F}(s(i))>\mathcal{F}(r(i))$, 
which is equivalent to $s(i)>\mathcal{F}^{-1}(\mathcal{F}(r(i))/\rho)$. Therefore, the domain of integration for calculating the expectation of $(1-d_i)\log_2(1+s(i))$ is $s(i)>\mathcal{F}^{-1}(\mathcal{F}(r(i))/\rho)$  and  $r(i)>0$, which leads to
the left hand side of (\ref{resenija}). Using a similar approach, the right hand side of (\ref{resenija}) is obtained from the right hand side of (\ref{max-solution}). This concludes the proof.
\end{proof}
\begin{remark}\label{remark4}
Eq.~(\ref{resenija}) reveals that the optimal decision threshold, $\rho_{\rm opt}$, depends indeed on the statistical properties of both involved links as was already alluded to in Remark \ref{remark3}. 
\end{remark}
\subsection{Rayleigh Fading\label{sec-Rayleigh}}
For concreteness, we provide in this subsection expressions for $\rho_{\rm opt}$ and the corresponding maximum throughput $\tau_{\max}$ for
 Rayleigh fading links. Thereby, the optimal and a suboptimal decision function $\mathcal{F}(\cdot)$ are considered.
\subsubsection{$\mathcal{F}(x)=x$} In this case, the limits $G(r)$ and $H(s)$ in (\ref{resenija}) are given by $G(r)=r/\rho$ and $H(s)=\rho s$. Thus, after some elementary manipulations, (\ref{resenija}) simplifies to
\begin{eqnarray}\label{find_rho}
 &&\hspace{-7mm} \exp\left(\frac{1}{  \Omega_S } \right) E_1\left(\frac{1}{  \Omega_S } \right) -\exp\left(\frac{1}{  \Omega_R }\right) E_1 \left(\frac{1}{  \Omega_R }\right) 
-\frac{ \Omega_R}{\Omega_R+\rho \Omega_S} 
\exp\left(\frac{\Omega_R+\rho \Omega_S}{ \Omega_S \Omega_R}\right) E_1
\left(\frac{\Omega_R+\rho \Omega_S}{ \Omega_S \Omega_R} \right)\nonumber\\
&&\quad+ \frac{\rho \Omega_S}{\Omega_R+\rho \Omega_S} 
\exp\left(\frac{\Omega_R+\rho \Omega_S}{ \rho \Omega_S \Omega_R} \right)
E_1 \left(\frac{\Omega_R+\rho \Omega_S}{ \rho \Omega_S \Omega_R} \right)=0.
\end{eqnarray}
The optimal value of $\rho$ for $\mathcal{F}(x)=x$, $\rho_{{\rm opt},1}$, can be obtained  from (\ref{find_rho}) via a simple one dimensional search.

The corresponding maximal throughput can be obtained from $\tau_{\max, 1}=E\{d_i \log_2(1+r(i))\}$ and is given by
\begin{eqnarray}\label{eq_to}
    \tau_{\max, 1}=\frac{1}{\ln(2)} \exp\left(\frac{1}{  \Omega_R }\right) E_1 \left(\frac{1}{  \Omega_R }\right)
    -\frac{1}{\ln(2)}\frac{\rho \Omega_S}{\Omega_R+\rho \Omega_S} 
\exp\left(\frac{\Omega_R+\rho \Omega_S}{ \rho \Omega_S \Omega_R} \right)
E_1 \left(\frac{\Omega_R+\rho \Omega_S}{ \rho \Omega_S \Omega_R} \right),
\end{eqnarray}
where $\rho=\rho_{{\rm opt},1}$.
\subsubsection{$\mathcal{F}(x)=\log_2(1+x)$} In this case, we have $G(r)=(1+r)^{\frac{1}{\rho}}-1$ and $H(s)=(1+s)^\rho-1$ in (\ref{resenija}). Thus, after some manipulations, we obtain from
(\ref{resenija}) 
\begin{eqnarray}\label{find_rho_log}
 &&\hspace{-8mm}\int_0^\infty\bigg[ \exp\left(-\frac{(r+1)^{\frac{1}{\rho }}-1}{\Omega_S}\right)\ln \left((r+1)^{\frac{1}{\rho }}\right)
+e^{\frac{1}{\Omega_S}} E_1\left(\frac{(r+1)^{\frac{1}{\rho }}}{\Omega_S}\right) 
\bigg]
\times\frac{1}{\Omega_R}\exp\left(-\frac{r}{\Omega_R}\right) dr \nonumber\\
&&\hspace{-8mm}-\int_0^\infty\bigg[
\exp\left(-\frac{(s+1)^{\rho }-1}{\Omega_R}\right) \ln \left((s+1)^{\rho }\right)
+e^{\frac{1}{\Omega_R}} E_1\left(\frac{(s+1)^{\rho }} {\Omega_R}\right)
\bigg] \times\frac{1}{\Omega_S}\exp\left(-\frac{s}{\Omega_S}\right) ds=0.
\end{eqnarray}
The optimal $\rho$, $\rho_{{\rm opt},2}$, can be found numerically from (\ref{find_rho_log}). 
The corresponding maximum throughput is obtained as 
\begin{eqnarray}
    \tau_{\max,2}=\frac{1}{\ln(2)}\int_0^\infty\bigg[
\exp\left(-\frac{(s+1)^{\rho }-1}{\Omega_R}\right)
\times \ln \left((s+1)^{\rho }\right)
+e^{\frac{1}{\Omega_R}} E_1 \left(\frac{(s+1)^{\rho }} 
{\Omega_R}\right)
\bigg]\frac{1}{\Omega_S}\exp\left(-\frac{s}{\Omega_S}\right) ds,
\label{eq_to_log}
\end{eqnarray}
where $\rho=\rho_{{\rm opt},2}$.
\subsubsection{Special case ($\Omega_S=\Omega_R$)} For the special case $\Omega_S=\Omega_R=\Omega$, we obtain from (\ref{find_rho}) and (\ref{find_rho_log}) $\rho_{{\rm opt},1}=\rho_{{\rm opt},2}=1$, 
and the corresponding maximal throughput is 
\begin{eqnarray}\label{eq_to_r1}
    \tau_{\max} = \tau_{\max, 1}=  \tau_{\max, 2}=\frac{1}{\ln(2)} \exp\left(\frac{1}{\Omega}\right) E_1 \left(\frac{1}{  \Omega}\right)
    -\frac{1}{2\ln(2)} 
\exp\left(\frac{2}{  \Omega} \right)
E_1 \left(\frac{2}{ \Omega} \right).
\end{eqnarray}
Comparing this throughput with the throughput achievable with conventional buffer-aided relaying without adaptive link selection, cf.~(\ref{th-conv}), the gain of adaptive link selection can be characterized by
\beq
\tau_{\max}/\tau_{\rm conv,2}  = 2- \exp\left(\frac{2}{  \Omega} \right)E_1 \left(\frac{2}{ \Omega} \right)/\left[\exp\left(\frac{1}{  \Omega} \right)E_1 \left(\frac{1}{ \Omega} \right)\right]\ge 1,
\eeq
where the ratio $\tau_{\max}/\tau_{\rm conv,2}$ monotonically increases from 1 to 1.5 as $\Omega$ decreases from $\infty$ to zero. We note that the results in Section \ref{numerics} reveal that the gain of adaptive link
selection is minimum for $\Omega_S=\Omega_R$, cf.~Fig.~\ref{fig2}.
\begin{remark}
From (\ref{kraj}) it is easy to see that for $\rho=1$ the decision functions $\mathcal{F}(x)=\log_2(1+x)$ and $\mathcal{F}(x)=x$ are equivalent in the sense that they lead to the same decisions. Hence, both decision functions lead
to identical throughputs $\tau_{\max}=\tau_{\max, 1}=\tau_{\max, 2}$. However, for $\Omega_S\ne\Omega_R$, $\rho_{{\rm opt},1}\ne\rho_{{\rm opt},2}\ne1$ holds, and the decision functions are no longer equivalent and $\tau_{\max, 2}>\tau_{\max, 1}$ holds.
\end{remark}
\section{Power Allocation\label{s34}}
So far, we have assumed that the source and relay transmit powers are fixed. In this section, we jointly optimize the power allocation and link selection policies for buffer-aided relaying.
 \subsection{Problem Formulation and Optimal Power Allocation}
Our goal is to jointly optimize the link selection variable $d_i$ and the powers $P_S(i)$ and $P_R(i)$ in each time slot $i$ such that the throughput is maximized. For convenience, we optimize in the
following the transmit  SNRs without fading $\gamma_S(i)$ and $\gamma_R(i)$, which may be viewed as normalized powers, instead of the powers $P_S(i)=\gamma_S(i)\sigma_{n_S}^2$ and $P_R(i)=\gamma_R(i)\sigma_{n_R}^2$ 
themselves. For a fair comparison, we limit the average (normalized) power consumed by the source and the relay to $\Gamma$. This leads for $N\to\infty$ to the following optimization problem:
 \begin{eqnarray}\label{power-eq-3}
\begin{array}{ll}
 {\underset{\gamma_S(i)\geq 0, \gamma_R(i)\geq 0, d_i}{\rm{Maximize: }}}&\frac{1}{N}\sum_{i=1}^N d_i \log_2(1+\gamma_R(i) h_R(i))\\
{\rm{Subject\;\; to: }} &{\rm C1:\;\;}\frac{1}{N}\sum_{i=1}^N (1-d_i) \log_2(1+\gamma_S(i) h_S(i))=\frac{1}{N}\sum_{i=1}^N d_i \log_2(1+\gamma_R(i) h_R(i))\\
& {\rm C2:\;\;}\frac{1}{N}d_i(1-d_i)=0\\
 & {\rm C3:\;\;}\frac{1}{N}\sum_{i=1}^N (1-d_i) \gamma_S(i)+\frac{1}{N}\sum_{i=1}^N d_i \gamma_R(i)\leq \Gamma
\end{array}
\end{eqnarray} 
where constraints C1 and C2 are identical to the constraints in (\ref{MPR1}) and C3 is the joint source-relay power constraint. The solution of Problem (\ref{power-eq-3}) is summarized in the following theorem. 
\begin{theorem}\label{theorem4}
The optimal (normalized) powers $\gamma_S(i)$ and $\gamma_R(i)$ and decision variable $d_i$ maximizing the throughput of buffer-aided relaying with adaptive link selection while satisfying an average 
source-relay power constraint are given  by
\begin{eqnarray}
    \gamma_S(i)&=&
\left\{
\begin{array}{cl}
 \rho/\lambda- 1/h_S(i)  &\textrm{ if } h_S(i)> \lambda/\rho  \\
0&\textrm{ otherwise }
\end{array}
\right.
\label{power-eq-2a}\\
\gamma_R(i)&=&
\left\{
\begin{array}{ll}
 1/\lambda - 1/h_R(i)  &\textrm{ if } h_R(i)> \lambda \;\;\; \\
0&\textrm{ otherwise }
\end{array}
\right.
\label{power-eq-2b}
\end{eqnarray}
\begin{eqnarray}\label{sel-var-d}
 d_i=
\left\{
\begin{array}{lll}
1 &{\rm if\;} \Big[\ln\left(\frac{h_R(i)}{\lambda}\right)+\frac{\lambda}{h_R(i)}-1 > \rho\ln\left(\frac{\rho}{\lambda}h_S(i)\right)+\frac{\lambda}{h_S(i)}-\rho
\textrm{ AND }   h_R(i)>\lambda \textrm{ AND } h_S(i)>\frac{\lambda}{\rho} \Big]
\\
 & \quad\textrm{ OR } \Big[ h_R(i)>\lambda \textrm{ AND } h_S(i)\leq\frac{\lambda}{\rho}\Big]\\
0 &{\rm otherwise}
\end{array}
\right.\label{kraj-pa}
\end{eqnarray}
    %
where the optimal $\rho$, $\rho_{\rm opt}$, and the optimal $\lambda$, $\lambda_{\rm opt}$, have to satisfy C1 and C3 in (\ref{power-eq-3}) for $N\to\infty$ with equality.
\end{theorem}
\begin{proof}
Please refer to Appendix \ref{app_pa}.
\end{proof}
 \begin{remark}
 Similar to the case without power allocation discussed in Section \ref{var-r-sol}, it is also possible for the case with power allocation to simplify the link selection policy in (\ref{kraj-pa}). For example, a simple suboptimal link selection policy which only depends
 on the instantaneous channel gains $h_S(i)$ and $h_R(i)$ but is independent of the transmit powers $\gamma_S(i)$ and $\gamma_R(i)$ may be adopted at the expense of some loss in performance. However, in this paper, we do not pursue suboptimal link 
 selection policies for the case of power allocation because of space limitation. 
\end{remark}
\subsection{Finding the Optimal $\lambda$ and $\rho$}
 The following  lemma establishes two equations from which the optimal $\lambda$ and $\rho$ can be found. 
\begin{lemma}\label{lemma2}
Denote the pdfs  of $h_S(i)$ and $h_R(i)$ by $f_{h_S}(h_S)$ and $f_{h_R}(h_R)$, respectively.  Let the transmit powers of the source and the relay in time slot $i$ be given by  (\ref{power-eq-2a}) and (\ref{power-eq-2b}), respectively, and the
link selection variable $d_i$ by (\ref{sel-var-d}). Then,  the optimal $\rho$,  $\rho_{\rm opt}$, and the optimal $\lambda$, $\lambda_{\rm opt}$, maximizing the throughput of buffer-aided relaying with adaptive link selection and power allocation
have to fulfill the following two equations 
\begin{eqnarray}\label{lemma-2-eq}
&&\hspace{-8mm}\int_{0}^{\lambda} \left[\int_{\lambda/\rho}^{\infty}  \log_2\left(\frac{\rho h_S}{\lambda}\right) f_{h_S}(h_S) d h_S\right] f_{h_R}(h_R) d h_R
+
\int_{\lambda}^\infty \left[\int_{L_1}^{\infty}  \left(\frac{\rho h_S}{\lambda}\right) f_{h_S}(h_S) d h_S\right] f_{h_R}(h_R) d h_R 
\nonumber\\
&&\hspace{-8mm}
= 
\int_{0}^{\lambda/\rho } \left[\int_{\lambda}^{\infty }  \log_2\left(\frac{h_R}{\lambda}\right) f_{h_R}(h_R) d h_R\right] f_{h_S}(h_S) d h_S
+
\int_{\lambda/\rho}^{\infty } \left[\int_{L_2}^{\infty }  \log_2\left(\frac{h_R}{\lambda}\right) f_{h_R}(h_R) d h_R\right] f_{h_S}(h_S) d h_S\; ,\nonumber\\
\end{eqnarray}
\begin{eqnarray}\label{lemma-2-eq-nova}
&&\hspace{-8mm}\int_{0}^{\lambda} \left[\int_{\lambda/\rho}^{\infty}  \left(\frac{\rho}{\lambda}-\frac{1}{h_S}\right) f_{h_S}(h_S) d h_S\right] f_{h_R}(h_R) d h_R
+
\int_{\lambda}^\infty \left[\int_{L_1}^{\infty}  \left(\frac{\rho}{\lambda}-\frac{1}{h_S}\right) f_{h_S}(h_S) d h_S\right] f_{h_R}(h_R) d h_R 
\nonumber\\
&&\hspace{-8mm}
+ 
\int_{0}^{\lambda/\rho } \left[\int_{\lambda}^{\infty } \left(\frac{1}{\lambda}-\frac{1}{h_R}\right) f_{h_R}(h_R) d h_R\right] f_{h_S}(h_S) d h_S
+
\int_{\lambda/\rho}^{\infty } \left[\int_{L_2}^{\infty }  \left(\frac{1}{\lambda}-\frac{1}{h_R}\right) f_{h_R}(h_R) d h_R\right] f_{h_S}(h_S) d h_S \nonumber\\
&&\hspace{-8mm}=\Gamma
\end{eqnarray}
where 
\begin{eqnarray}
 L_1&=&-\frac{\lambda }{\rho W(-e^{(h_R-\lambda)/(\rho h_R)-1}(\lambda/h_R)^{1/\rho})}, \quad  L_2=-\frac{\lambda }{W(-e^{\rho-1-\lambda/h_S}(\lambda/(\rho h_S))^{\rho})}.
\label{L_1}
\end{eqnarray}
Here, $W(\cdot)$ is the Lambert $W$-function \cite{corless1996lambertw}, which is available as built-in function in software packages such as Mathematica. The maximum throughput is given by the left  (and right) hand side of (\ref{lemma-2-eq}).
\end{lemma}
\begin{proof}
Please refer to Appendix \ref{proof_of_lemma_2}.
\end{proof}
To find $\lambda_{\rm opt}$ and $\rho_{\rm opt}$ meeting (\ref{lemma-2-eq}) and (\ref{lemma-2-eq-nova}), a two dimensional search over all $\lambda>0$ and $\rho>0$ has to be conducted or built-in root-finding functions of commercially available software such as Mathematica can be used.
The optimal $\lambda$ and $\rho$ can be found offline since (\ref{lemma-2-eq}) and (\ref{lemma-2-eq-nova}) only depend on the statistical properties of the $\mathcal{S}$-$\mathcal{R}$ and the $\mathcal{R}$-$\mathcal{D}$ links. Since these statistical properties change on a much slower time scale than
 the instantaneous channel gains, $\lambda_{\rm opt}$ and $\rho_{\rm opt}$ can be updated with a low rate.
\subsection{Rayleigh Fading}\label{sec-th-pa-1} 
For the  special case of Rayleigh fading with $f_{h_S}(h_S)=e^{-h_S/\bar\Omega_S}/\bar\Omega_S$ and $f_{h_R}(h_R)=e^{-h_R/\bar\Omega_R}/\bar\Omega_R$, (\ref{lemma-2-eq})
can be simplified to
 \begin{eqnarray}\label{ra1}
&&\frac{1}{\ln(2)} \bigg[
\left(1-e^{-\lambda/\bar\Omega_R}\right) E_1\left(\frac{\lambda}{\rho \bar\Omega_S}\right)
+ \int_{\lambda}^\infty
\bigg\{e^{-L_1/\bar\Omega_S}\ln\left(\frac{\rho L_1}{\lambda}\right)+ E_1\left(\frac{L_1}{\bar\Omega_S}\right)\bigg\}\frac{e^{-h_R/\bar\Omega_R}}{\bar\Omega_R} d h_R\bigg]\nonumber\\
 &&=
 \frac{1}{\ln(2)} \bigg[
\left(1-e^{-\lambda/(\rho\bar\Omega_S)}\right) E_1\left(\frac{\lambda}{\bar\Omega_R}\right)
+ \int_{\lambda/\rho}^\infty
\bigg\{e^{-L_2/\bar\Omega_R}\ln\left(\frac{L_2}{\lambda}\right)+ E_1\left(\frac{L_2}{\bar\Omega_R}\right)\bigg\}\frac{e^{-h_S/\bar\Omega_S}}{\bar\Omega_S} d h_S\bigg]  ,
\end{eqnarray}
and (\ref{lemma-2-eq-nova}) can be simplified to
\begin{eqnarray}\label{ra2}
&&
\left(1-e^{-\lambda/\bar\Omega_R}\right)\left(\frac{\rho}{\lambda}  e^{-\lambda/(\rho\bar\Omega_S)}-\frac{E_1\left(\frac{\lambda}{ \rho \bar\Omega_S}\right)}{\bar\Omega_S}\right)
+ \int_{\lambda}^\infty
\left\{\frac{\rho}{\lambda}  e^{-L_1/ \bar\Omega_S}-\frac{E_1\left(\frac{L_1}{ \bar\Omega_S}\right)}{\bar\Omega_S}\right\}\frac{e^{-h_R/\bar\Omega_R}}{\bar\Omega_R} d h_R\nonumber\\
 &&+
 \left(1-e^{-\lambda/(\rho\bar\Omega_S)}\right)\left(\frac{1}{\lambda}  e^{-\lambda/\bar\Omega_R}-\frac{E_1\left(\frac{\lambda}{  \bar\Omega_R}\right)}{\bar\Omega_R}\right)
+ \int_{\lambda/\rho}^\infty
\left\{\frac{1}{\lambda}  e^{-L_2/ \bar\Omega_R}-\frac{E_1\left(\frac{L_2}{ \bar\Omega_R}\right)}{\bar\Omega_R}\right\}\frac{e^{-h_S/\bar\Omega_S}}{\bar\Omega_S} d h_S
=\Gamma,\qquad
\end{eqnarray}
where $L_1$ and $L_2$ are given in (\ref{L_1}) and the maximum throughput is given by the left (and right) hand side of equation (\ref{ra1}).
\section{Delay Limited Transmission}\label{sec-delay}
So far, we have assumed that there is no delay constraint and that the size of the buffer at the relay is infinite. In practice, there is usually some constraint on the delay and on the buffer size.
In this section, we investigate how these constraints affect the performance of the proposed relaying protocol. For simplicity, we assume fixed transmit powers, i.e., $P_S(i)=P_S$, $P_R(i)=P_R$, $\forall i$, and
consequently, policy (\ref{kraj}) is used for link selection.

Since we assume that the source is backlogged and has always information to transmit, for the considered three-node network, the transmission delay is caused only by the buffer at the relay. Let $T(i)$ denote the 
delay of a bit that is transmitted by the source in time slot $i$ and received at the destination in time slot $i+T(i)$, i.e., the considered bit is stored for $T(i)$ time slots in the buffer. Then, according to Little's law \cite{little} 
the average delay $E\{T(i)\}$ (i.e., the average time that a bit is stored in the buffer) is given by
\begin{eqnarray}\label{delay-main}
    E\{T(i)\}=E\{Q(i)\}/A,
\end{eqnarray}
where $E\{Q(i)\}$ is the average queue length at the buffer and $A$ is the average arrival rate into the queue. $E\{Q(i)\}$ is given in bits and $A$ is given in bits/slot. Thus, the average delay $E\{T(i)\}$ is given in time slots.
From (\ref{delay-main}), we observe that the delay can be controlled via the arrival rate and the queue size. In the following, we will present two different approaches to adjust the arrival rate and the queue size. The first approach is to ``starve" the buffer, i.e., we 
intentionally limit the arrival rate by choosing $\rho<\rho_{\rm opt}$. The second approach is to limit the buffer size by forcing the relay to transmit if the buffer gets full. 
\subsection{Satisfying an Average Delay Constraint by ``Starving" the Buffer}\label{sec-delay-1}
Starving the buffer is a common approach for limiting average delays in queueing systems \cite{neely-2,gross-harris}. In our case, we can decrease the average arrival rate by selecting $\rho<\rho_{\rm opt}$ which leads to $E\{(1-d_i) S(i)\}<E\{d_i R(i)\}$. In the following
theorem, we establish an upper bound for the resulting average delay.
\begin{theorem}\label{theorem8}
Let $\rho<\rho_{\rm opt}$ in (\ref{kraj}) such that $\xi=E\{(1-d_i) S(i)\}/E\{d_i R(i)\}<1$. In this case, assuming slot-by-slot uncorrelated fading the average delay in slots is bounded by 
\begin{eqnarray}\label{eq-f1}
    E\{T(i)\}\leq \frac{1}{2}\frac{1}{E\{(1-d_i) S(i)\}}\frac{E\{(1-d_i) S^2(i)\}+\xi (2-\xi) E\{ d_i R^2(i) \} }{E\{d_i R(i)\} - E\{(1-d_i) S(i)\}}
\end{eqnarray}
and the throughput is given by $\tau=E\{(1-d_i) S(i)\}$.
\end{theorem}
\begin{proof}
Please refer to Appendix \ref{app_D}.
\end{proof}
\begin{remark}\label{remark9}
Exploiting  (\ref{eq-f1}) the required value of $\rho<\rho_{\rm opt}$ to ensure a desired average delay can be found. For example, assuming $\mathcal{F}(x)=x$ and Rayleigh distributed $\mathcal{S}$-$\mathcal{R}$ and $\mathcal{R}$-$\mathcal{D}$
channel gains the expected values required in (\ref{eq-f1}) can be obtained as
\begin{eqnarray}\label{eq-r11}
   &&\hspace{-14mm} E\{(1-d_i)S(i)\}= \frac{1}{\ln(2)}\left[ \exp\left(\frac{1}{  \Omega_S } \right) E_1\left(\frac{1}{  \Omega_S } \right) -
\frac{ \Omega_R}{\Omega_R+\rho \Omega_S} 
\exp\left(\frac{\Omega_R+\rho \Omega_S}{ \Omega_S \Omega_R}\right) E_1
\left(\frac{\Omega_R+\rho \Omega_S}{ \Omega_S \Omega_R} \right) 
\label{e1}\right]\\
 &&\hspace{-14mm} E\{d_iR(i)\}=\frac{1}{\ln(2)} \left[ \exp\left(\frac{1}{  \Omega_R }\right) E_1 \left(\frac{1}{  \Omega_R }\right)- \frac{\rho \Omega_S}{\Omega_R+\rho \Omega_S} 
\exp\left(\frac{\Omega_R+\rho \Omega_S}{ \rho \Omega_S \Omega_R} \right)
E_1 \left(\frac{\Omega_R+\rho \Omega_S}{ \rho \Omega_S \Omega_R} \right)
\label{e2}\right]\\
  &&\hspace{-14mm}  E\{(1-d_i) S^2(i)\}=\int_0^\infty\left[\int_{r/\rho}^\infty (\log_2(1+s))^2\frac{e^{-s/\Omega_S}}{\Omega_S} ds\right] \frac{e^{-r/\Omega_R}}{\Omega_R} dr
    \label{e3}\\
&&\hspace{-14mm} E\{d_i R^2(i)\}=\int_0^\infty\left[\int_{s \rho}^\infty (\log_2(1+r))^2\frac{e^{-r/\Omega_R}}{\Omega_R} dr\right] \frac{e^{-s/\Omega_S}}{\Omega_S} ds.\label{e4}
\end{eqnarray}
Inserting now (\ref{e1})-(\ref{e4}) into (\ref{eq-f1}), the value of $\rho$ guaranteeing a certain average delay can be found by slowly increasing $\rho$ from zero until the right hand side of  (\ref{eq-f1}) equals the desired average delay $E\{T(i)\}$. 
\end{remark}

If the buffer size is limited, there is a non-zero probability that the bits arriving into the buffer have to be dropped because the buffer is full, even if the buffer is starved. However, the probability of this event happening can be minimized by properly
choosing the buffer size compared to the desired average delay. This issue is addressed in the following lemma.
\begin{lemma}\label{lemma3}
Denote the maximum queue size by $Q_{\max}$. Then, we can bound the probability that the queue in the buffer exceeds $Q_{\max}$ as
\begin{eqnarray}\label{drop}
   {\rm Pr}\{Q(i)>Q_{\max}\}\leq E\{Q(i)\}/Q_{\max}.
\end{eqnarray} 
\end{lemma}
\begin{proof}
The proof follows directly from Markov's inequality.
\end{proof}
We can guarantee any prescribed probability of dropped bits, ${\rm Pr}\{Q(i)>Q_{\max}\}$, by selecting appropriate values for $\rho$ and $Q_{\max}$ based on (\ref{drop}) and (\ref{eq_d3}) given in Appendix \ref{app_D}. The resulting throughput is given by 
\begin{eqnarray}
    \tau&=&E\{(1-d_i) S(i)|Q(i)<Q_{\max}\}  {\rm Pr}\{Q(i)<Q_{\max}\}
=E\{(1-d_i) S(i)\}  {\rm Pr}\{Q(i)<Q_{\max}\}\nonumber\\
&=&E\{(1-d_i) S(i)\} (1-  {\rm Pr}\{Q(i)>Q_{\max}\}),
\label{tau}
\end{eqnarray}
for which a lower bound can be found by combining (\ref{drop}), (\ref{tau}), and (\ref{eq_d3}).
\subsection{Satisfying the Delay Constraint by Limiting the Queue Size}\label{sec-delay-2}
For the scheme proposed in the previous subsection dropped bits are unavoidable. In this subsection, we propose an alternative approach which allows us to avoid dropped bits. Let the buffer size again be limited to $Q_{\max}$ bits. 
The proposed scheme employs the following link selection protocol:
\begin{enumerate}
\item  If  $Q_{\max}-Q(i-1)>S(i)$, select $d_i$ based on  (\ref{kraj}).
\item Otherwise, set $d_i=1$.
\end{enumerate}
Hence, if there is enough room in the buffer to accommodate the bits possibly sent from the source to the relay, the link selection protocol introduced in Section \ref{var-r-sol} is employed. On the other hand,
if there exists the possibility of a buffer overflow, the relay transmits to reduce the amount of data in the buffer. 

\begin{remark}
Although conceptually simple, a theoretical analysis of the throughput of the queue size limiting protocol is difficult. In contrast to the buffer starving protocol discussed in Section \ref{sec-delay-1}, for the queue size limiting protocol,
the average arrival rate $A$, depends on the frequency with which the buffer has to be emptied due to a full queue.  The frequency of these events depends in turn on the average arrival rate. This mutual dependence of average arrival rate and emptying
the buffer makes a meaningful theoretical analysis difficult. Thus, we will resort to simulations to evaluate the performance of the queue size limiting protocol in Section \ref{numerics}.
\end{remark}
\begin{remark}
We note that both proposed protocols for the delay constrained case are heuristic in nature. The search for other protocols with possibly superior performance is an interesting topic for future work. The proposed protocols for the delay constrained and the delay
unconstrained case can serve as benchmark and performance upper bound for these new protocols, respectively. 
\end{remark}
\section{Numerical and Simulation Results}\label{numerics}
In this section, we evaluate the performance of buffer-aided DF relaying with adaptive link selection and compare it with that of conventional relaying. Throughout this section, we assume  Rayleigh fading. All results shown in this section have been confirmed by computer simulations. However,
the simulations are not shown in all instances for clarity of presentation.
\subsection{Delay Unconstrained Transmission}
First, we assume that there are no delay constraints and investigate the achievable throughputs with and without power allocation. 
\subsubsection{Throughput of Buffer-Aided Relaying with Adaptive Link Selection}
We first consider the case of fixed transmit powers. In Fig.~\ref{fig2}, we show the ratio of the optimal throughput of buffer-aided relaying with adaptive link selection, $\tau_{\rm max}$, and the throughput of conventional relaying with a buffer, $\tau_{\rm conv,2}$, given in (\ref{th-conv}), 
as a function of $\Omega_R/\Omega_S$ for several different values of $\Omega_S$. The corresponding optimal decision thresholds, $\rho_{\rm opt}$, for buffer-aided relaying with adaptive link selection are shown in Fig.~\ref{fig3}.
For buffer-aided relaying with adaptive link selection,  we considered the decision functions $\mathcal{F}(x)=x$ and $\mathcal{F}(x)=\log_2(1+x)$ and calculated the corresponding throughputs based on (\ref{eq_to}) and (\ref{eq_to_log}), respectively. 
The optimal decision thresholds were obtained from (\ref{find_rho}) and (\ref{find_rho_log}), respectively. Clearly, buffer-aided relaying with adaptive link selection leads to substantial throughput gains compared to conventional relaying. Both considered decision functions lead  
to very similar performances, although at very high ratios $\Omega_R/\Omega_S$ the optimal decision function $\mathcal{F}(x)=\log_2(1+x)$ yields a small throughput gain. The ratio $\tau_{\rm max}/\tau_{\rm conv,2}$ approaches two as $\Omega_R/\Omega_S\to 0$ 
and $\Omega_R/\Omega_S\to \infty$. For $\Omega_R/\Omega_S\to 0$, the source-relay link is selected very rarely for transmission (as $\rho_{\rm opt}\to 0$) since comparatively large amounts of data can be transferred to the relay in a single time slot. 
Thus, the relay can almost always transmit as compared to half of the time in conventional relaying. On the other hand, for $\Omega_R/\Omega_S\to \infty$, it is the relay-destination channel that is used very rarely as $\rho_{\rm opt}\to \infty$
and the source can transmit almost all the time, which results in twice the throughput compared to conventional relaying.
\subsubsection{Throughput with Power Allocation}
In Figs.~\ref{fig4b} and \ref{fig4c}, we investigate the gains achievable with power allocation (PA) for a system with $\bar\Omega_S=0.1$ and $\bar\Omega_R=1.9$. Thereby, we compare the performances of buffer-aided relaying with adaptive 
link selection and conventional relaying with and without a buffer. For buffer-aided relaying with adaptive link selection and power allocation the throughput, power allocation, and link selection policy were obtained as described in Theorem \ref{theorem4}
and Lemma \ref{lemma2} in Section \ref{s34}. For conventional relaying with buffer, a similar optimal power allocation scheme as for buffer-aided relaying with adaptive link selection was adopted with variable powers $\gamma_j(i)=\max\{0,1/\alpha-1/ h_j(i)\}$, $j\in\{S,R\}$, 
where $\alpha$ is chosen such that $E\{\gamma_j(i)\}=\Gamma$. For comparison, in Figs.~\ref{fig4b} and \ref{fig4c}, we also
show the performance of both considered relaying schemes without power allocation, i.e., we set $\gamma_S(i)=\gamma_R(i)=\Gamma$. Furthermore, to highlight the gain compared to conventional relaying without buffer and without power allocation,
in Fig.~\ref{fig4b}, we normalized the throughput with respect to  $\tau_{\rm conv,1}$ as given by (\ref{t-c-1}). Figs.~\ref{fig4b} and \ref{fig4c} show that optimal power allocation can improve performance of both buffer-aided relaying with adaptive
link selection and conventional relaying. For example, for $\Gamma=0$ dB  buffer-aided relaying with adaptive link selection and power allocation leads to a throughput gain of 95 \% compared to conventional relaying with buffer and power allocation. 
Nevertheless, the gain achievable by adaptive link selection is more
significant than the gain from power allocation. For example, at $\Gamma=20$ dB, adaptive link selection yields a throughput gain of 1 bit/slot compared to conventional relaying with buffer. 
\subsection{Delay Constrained Transmission}
We now turn our attention to delay limited transmission and investigate the performance of the two proposed protocols for this case. In the following we assume uncorrelated fading. Furthermore, we assume fixed transmit powers for the source and relay, and adopt the suboptimal decision function $\mathcal{F}(x)=x$.
\subsubsection{Starving the Buffer}
In Fig.~\ref{fig5}, we show the ratio of the  throughput of buffer-aided relaying with starved buffer, $\tau_{\rm max}$,  and the throughput of conventional relaying without buffer, $\tau_{\rm conv,1}$, as given in (\ref{t-c-1}),  
as a function of the upper bound on the average delay. The corresponding decision thresholds, $\rho$, are shown in Fig.~\ref{fig6}. For the theoretical results shown in Figs.~\ref{fig5} and \ref{fig6} the throughput and the upper bound 
on the average delay for buffer-aided relaying with starved buffer were obtained based on Theorem \ref{theorem8} and Remark \ref{remark9}. A comparison of the theoretical results with the simulation results also shown in Fig.~\ref{fig5} 
reveals that the derived upper bound on the delay is tight, especially for large delays. On the other hand, for very small delays, buffer-aided relaying with starved buffer becomes inefficient since the starving of the buffer decreases
the average queue size and increases the probability that the relay is selected for transmission when  $R(i)>Q(i-1)$. In fact, Fig.~\ref{fig5} shows that for very small delays buffer-aided relaying with starved buffer may be even outperformed
by  conventional relaying without buffer.
As expected, the throughput of buffer-aided relaying increases with increasing tolerable delay and $\rho$ approaches the optimal value for the
delay constrained case, $\rho_{\rm opt}$ (computed from (\ref{find_rho})), for large delays.
The required delay to achieve a throughput gain compared to conventional relaying increases with increasing $\Omega_S$-$\Omega_R$ ratio since for large $\Omega_S/\Omega_R$ and small tolerable average delays, the arrival rate into 
the buffer has to be severely limited (i.e, $\rho$ has to be chosen very small) which has a negative impact on the throughput which, by the conservation of flow, is equal to the arrival rate. 

In Fig.~\ref{fig7}, we show the probability of a dropped bit as a function of the buffer size $Q_{\rm max}$ for three different average delays and $\Omega_S=\Omega_R=1$. These results were obtained via simulations since the bound
obtained in (\ref{drop}) is relatively loose due to the looseness of Markov's inequality. Fig.~\ref{fig7} shows that the probability of dropping a bit rapidly decreases with increasing buffer size and decreasing average delay. 
\subsubsection{Limiting the Queue Size} 
In Fig.~\ref{fig8}, we show the throughput achieved by limiting the queue size, $\tau_{\rm limit}$, and the throughput achieved by starving the buffer, $\tau_{\rm starve}$, normalized by the throughput of conventional relaying without buffer, $\tau_{\rm conv,1}$, 
as given in (\ref{t-c-1}), for symmetric and asymmetric link qualities. For comparison, Fig.~\ref{fig8} also contains the throughput of delay unconstrained buffer-aided relaying with adaptive link selection, which constitutes an upper bound for the throughput in the
delay constrained case, and conventional relaying with buffer as proposed in \cite{XFTP08}. For conventional relaying with buffer, the relay drops information bits if the achievable rate of the $\cal S$-$\cal R$ link in the first $N/2$ time slots exceeds the achievable 
rate of the $\cal R$-$\cal D$ link in the second $N/2$ time slots.
All results shown in Fig.~\ref{fig8} have been obtained by simulations. Fig.~\ref{fig8} reveals that the performance of both delay constrained buffer-aided relaying protocols 
is comparable for large delays and approaches that of the delay unconstrained protocol. For small delays, limiting the buffer size yields a higher throughput than starving the buffer. However, both proposed protocols may be outperformed by conventional relaying 
with and without buffer for very small delays as for the proposed simple protocols, the relay may be selected for transmission even if $R(i)>Q(i-1)$. While this event has negligible effect for delay unconstrained transmission since the optimal link selection policy ensures that the queue is sufficiently long such that $R(i)>Q(i-1)$ is avoided (cf.~Remark \ref{remark1}), this is no longer true for the delay constrained case. Therefore, for the delay constrained case, more sophisticated protocols should be developed that take into account that $R(i)>Q(i-1)$ may occur.

\section{Conclusions}\label{conclude}
In this paper, we proposed a novel adaptive link selection protocol for relays with buffers. In contrast to conventional relaying, where the source and the relay transmit according to a pre-defined schedule regardless of the channel state, in the proposed scheme,
always the node with the stronger link is selected for transmission. For delay unconstrained transmission, we derived the optimal link selection policy for the cases of fixed and variable source and relay transmit powers. Remarkably, in both cases, the optimal 
policy for a given time slot only depends on the instantaneous CSI of that time slot and the statistical CSI of the involved links. This makes the optimal policies attractive for implementation. For delay constrained transmission, we proposed two different methods to 
control the delay introduced by the buffer at the relay. Furthermore, for the case when the buffer is starved, we derived upper bounds on the average delay and the number of dropped bits for limited buffer size. Our analytical and simulation results showed 
that buffer-aided relaying with adaptive link selection is a promising approach to significantly increase the throughput compared to conventional relay-assisted transmission. 
Interesting extensions of the presented work include using the considered simple three-node network as a building block for larger networks, studying the impact of imperfect CSI, and deriving the outage probability for fixed rate transmission.
\newpage
\begin{appendix}
\subsection{Proof of Theorem \ref{theorem1}\label{app_A}}
We first note that, because of the law of the conservation of flow, $A\ge \tau$ is always valid and equality holds if and only if the queue is non-absorbing. Assume first we have a link selection policy with average arrival rate $A$ and throughput $\tau$ with $A>\tau$, 
i.e., the queue is absorbing.  For this policy, we denote the set of indices with $d_i=1$ by $\bar I$ and the set of indices with $d_i=0$ by $I$, i.e., for $N\to \infty$ we have
\beq
A=\frac{1}{N}\sum_{i\in I}(1-d_i)S(i)>\tau =\frac{1}{N}\sum_{i\in\bar I}d_i\min\{R(i),Q(i-1)\}.
\label{app1}
\eeq
From (\ref{app1}) we observe that the considered protocol cannot be optimal as it can be improved by moving some of the indices $i$ in $I$ to $\bar I$ which leads to an increase of $\tau$ at the expense of a decrease of $A$. However, once the point $A=\tau$ is reached,
moving more indices $i$ from $I$ to $\bar I$ will decrease both $A$ and $\tau$  because of the conservation of flow. Thus, a necessary condition for the optimal policy is that the queue is non-absorbing but the queue is at the edge of non-absorbtion, i.e., the queue is 
at the boundary of a non-absorbing and an absorbing queue. This completes the proof.
\subsection{Proof of Theorem \ref{theorem2}\label{app_B}}
We denote the sets of indices $i$ for which $d_i=1$ and $d_i=0$ holds by $\bar I$ and $I$, respectively. $\epsilon$ denotes a subset of $\bar I$ and $|\cdot|$ is the cardinality of a set. Throughout the remainder of this proof $N\to\infty$ is assumed.

If the queue in the buffer of the relay is absorbing, $A>\tau$ holds and on average the number of bits arriving at the queue exceed the number of bits leaving the queue. Thus, $R(i)\le Q(i-1)$ holds almost always and as a result the throughput can be written as 
\begin{eqnarray}\label{pr_eq_2}
   \tau= \frac{1}{N}\sum_{i\in \bar I }\min\{R(i),Q(i-1)\} = \frac{1}{N} \sum_{i\in \bar I} R(i) .
\end{eqnarray}

Now,  we assume that the queue is at the edge of non-absorption. That is $A=\tau$ holds but  moving a small fraction $\epsilon$, where $|\epsilon|/N\to 0$,  of indices from $\bar I$ to $I$ will make  the queue an absorbing queue with $A>\tau$. 
For this case, we wish to determine whether or not
\begin{eqnarray}\label{pr_eq_3}
   \frac{1}{N} \sum_{i\in \bar I} R(i)> \tau =\frac{1}{N}\sum_{i\in \bar I }\min\{R(i),Q(i-1)\}= A=\frac{1}{N}\sum_{i\in   I } S(i)
\end{eqnarray}
holds. To test this, we move a small fraction $\epsilon$, where $|\epsilon|/N\to 0$, of indices from $\bar I$ to $I$, thus making the queue an absorbing queue. As a result, (\ref{pr_eq_2}) holds, and (\ref{pr_eq_3}) becomes
\begin{eqnarray}\label{pr_eq_4}
   \frac{1}{N} \sum_{i\in \bar I\backslash\epsilon} R(i)= \tau =\frac{1}{N}\sum_{i\in \bar I\backslash \epsilon }\min\{R(i),Q(i-1)\}< A= \frac{1}{N}\sum_{i\in I \cup \epsilon } S(i).
\end{eqnarray}
From the above we conclude that if (\ref{pr_eq_2}) holds, then based on (\ref{pr_eq_3}) and (\ref{pr_eq_4}), for $|\epsilon|/N\to 0$, we must have
\begin{eqnarray}\label{pr_eq_5}
   \frac{1}{N} \sum_{i\in\bar I} R(i)>  \frac{1}{N}\sum_{i\in   I } S(i)
\end{eqnarray}
and
\begin{eqnarray}\label{pr_eq_6}
   \frac{1}{N} \sum_{i\in \bar I\backslash\epsilon} R(i) <  \frac{1}{N}\sum_{i\in   I \cup \epsilon } S(i).
\end{eqnarray}
However, for (\ref{pr_eq_5}) and (\ref{pr_eq_6}) to jointly hold, we require that the particular considered moving of indices from $\bar I$ to $I$ has caused a discontinuity in $\frac{1}{N} \sum_{i\in \bar I} R(i)$ or/and a discontinuity in $\frac{1}{N}\sum_{i\in   I } S(i)$ as $|\epsilon|/N\to 0$
is assumed. Since the capacities of the $\mathcal{S}$-$\mathcal{R}$ and $\mathcal{R}$-$\mathcal{D}$ links are such that $\lim_{N\to\infty}\sum_{i\in\epsilon}S(i)/N\to 0$ and $\lim_{N\to\infty}\sum_{i\in\epsilon}R(i)/N\to 0$, such discontinuities are not possible. Therefore, at the edge of non-absorption 
(\ref{pr_eq_3}) is not true and we must have instead
\begin{eqnarray}\label{pr_eq_7}
   \frac{1}{N} \sum_{i\in \bar I} R(i)= \tau= \frac{1}{N}\sum_{i\in \bar I }\min\{R(i),Q(i-1)\}=A= \frac{1}{N}\sum_{i\in   I } S(i)
\end{eqnarray}
Using the the ergodicity of $s(i)$ and $r(i)$, (\ref{pr_eq_7}) can be expressed as (\ref{max-solution}), and the throughput  can be written as (\ref{max-tau}). This concludes the proof. 
\subsection{Proof of Theorem \ref{theorem3}\label{app_C}}
The Lagrangian for Problem (\ref{MPR1}) is given by
\begin{eqnarray}\label{MPR2}
    \mathcal{L}=\frac{1}{N}\sum_{i=1}^N d_i R(i) -\mu \frac{1}{N}\sum_{i=1}^N  \left[d_i R(i)-(1-d_i)S(i)\right] - \frac{1}{N}\sum_{i=1}^N \beta_i d_i(1-d_i),
\end{eqnarray}
where $\mu$ and $\beta_i$ are Lagrange multipliers. Differentiating $\mathcal{L}$ with respect to $d_i$ and setting the result to zero leads to
\begin{eqnarray}\label{eq-nekoja1}
    d_i=\frac{(-1+\mu)R(i)+\mu S(i)+\beta_i}{2\beta_i}.
\end{eqnarray}
For $d_i(1-d_i)=0$ to hold, we need either $d_i=0$ or $d_i=1$, which leads to two possible values for $\beta_i$:
\begin{eqnarray}
d_i=0\, \Rightarrow\,    \beta_{i,1}&=&(1-\mu)R(i)-\mu S(i)\\
d_i=1\, \Rightarrow\,  \beta_{i,2}&=&-\beta_{i,1}
\end{eqnarray}
For the maximum of $\mathcal{L}$ in (\ref{MPR2}), $\beta_i<0$, $\forall i$, has to hold. Furthermore, $0<\mu<1$ has to hold since for $\mu\le 0$ and $\mu\ge 1$ we have always $d_i=1$ and $d_i=0$, respectively, irrespective of the values of $R(i)$ and $S(i)$. 
Hence, we have
\begin{eqnarray}\label{link-select-1}
   d_i=\left\{
\begin{array}{cc}
1 & \textrm{if }   (1-\mu)R(i)-\mu S(i) \geq 0\\
0 & \textrm{if } (1-\mu)R(i)-\mu S(i)   < 0
\end{array} 
\right.
\end{eqnarray}
which is identical to (\ref{kraj}) with $\mathcal{F}(x)=\log_2(1+x)$ if we set $\rho=\mu/(1-\mu)$. $\mu$ or equivalently $\rho$ are chosen such that constraint C1 of Problem (\ref{MPR1}) is met. This completes the proof.
\subsection{Proof of Theorem \ref{theorem4}\label{app_pa}}
To solve Problem (\ref{power-eq-3}), we form the Lagrangian
\begin{eqnarray}\label{MPR2a}
    \mathcal{L}=&&\hspace{-5mm}\frac{1}{N}\sum_{i=1}^N d_i \log_2(1+\gamma_R(i) h_R(i)) 
-\mu \frac{1}{N}\sum_{i=1}^N \Big[d_i \log_2(1+\gamma_R(i) h_R(i))
-(1-d_i) \log_2(1+\gamma_S(i) h_S(i))\Big]\nonumber\\
-&&\hspace{-5mm} \nu \frac{1}{N}\sum_{i=1}^N \Big[ (1-d_i) \gamma_S(i)+d_i \gamma_R(i)\Big]
-\frac{1}{N}\sum_{i=1}^N \beta_i d_i(1-d_i),
\end{eqnarray}
where the Lagrange multipliers $\mu$, $\beta_i$, and $\nu$ are chosen such that C1, C2, and C3 are satisfied, respectively.
By differentiating $\mathcal{L}$ with respect to $\gamma_S(i)$,  $\gamma_R(i)$, and $d_i$, and setting the results to zero, we obtain three equations. 
Solving this system of equations for $\gamma_S(i)$,  $\gamma_R(i)$, and $d_i$, and taking into account that  $\beta_i<0$, $0<\mu<1$, and $\nu>0$, we
obtain  (\ref{power-eq-2a}), (\ref{power-eq-2b}), and (\ref{sel-var-d}) after letting $\rho=\mu/(1-\mu)$ and $\lambda=\nu\ln(2)/(1-\mu)$, which are chosen such that 
constraints C1 and  C3 are meet with equality. This completes the proof.
\subsection{Proof of Lemma \ref{lemma2}\label{proof_of_lemma_2}}
Since $s(i)$ and  $r(i)$ are ergodic random processes, for $N\to\infty$, the normalized sums in C1 and C3 in (\ref{power-eq-3}) can be replaced by expectations. 
Therefore, the left hand side of C1 is the expectation of variable $(1-d_i)\log_2(1+\gamma_S(i) h_S(i))$. This variable is nonzero only when both $(1-d_i)$ and $\gamma_S(i)$ are nonzero. 
The domain over which  $(1-d_i)$ and $\gamma_S(i)$ are jointly nonzero can be obtained from (\ref{power-eq-2a}) and (\ref{sel-var-d}) and is given by
\beq
(h_S(i)>\lambda/\rho \textrm{ AND } h_R(i)<\lambda) \textrm{ OR  } (h_S(i)>L_1 \textrm{ AND } h_R(i)>\lambda)
\label{dom1}
\eeq
where $L_1$  is given by (\ref{L_1}). Variable $(1-d_i)\log_2(1+\gamma_S(i) h_S(i))$ has to be integrated over domain (\ref{dom1}) to obtain its average. This leads to  the left side of (\ref{lemma-2-eq}).

Similarly, the right hand side of C1 is the expectation of the variable $d_i\log_2(1+\gamma_R(i) h_R(i))$. This variable is nonzero only when both $d_i$ and $\gamma_R(i)$ are nonzero. The domain over which  $d_i$ 
and $\gamma_R(i)$ are jointly nonzero can be obtained from (\ref{power-eq-2b}) and (\ref{sel-var-d}) and is given by
\beq
(h_R(i)>\lambda  \textrm{ AND } h_S(i)<\lambda/\rho) \textrm{ OR  } (h_R(i)>L_2 \textrm{ AND } h_S(i)>\lambda/\rho)
\label{dom2}
\eeq
where $L_2$  is given by (\ref{L_1}). Variable $d_i\log_2(1+\gamma_R(i) h_R(i))$ has to be integrated over domain (\ref{dom2}) to obtain its average. This leads to the right side of (\ref{lemma-2-eq}).

Following a similar procedure, we can obtain (\ref{lemma-2-eq-nova}) from  C3 in  (\ref{power-eq-3}). This completes the proof.
\subsection{Proof of Theorem \ref{theorem8}\label{app_D}}
For $\xi=E\{(1-d_i) S(i)\}/E\{d_i R(i)\}<1$ the queue is non-absorbing, and thus, because of the law of conservation of flow, the throughput is equal to the arrival rate, i.e., $\tau=E\{(1-d_i) S(i)\}$.

To arrive at an upper bound for the average queue size, we first introduce two auxiliary results from the literature. Let
\begin{eqnarray}\label{q-i}
    q(i)=\max\{q(i-1)-u(i),0\},
\end{eqnarray}
where $u(i)$ is a slot by slot uncorrelated random variable with $E\{u(i)\}>0$. Also, let $a(i)$ and $b(i)$ be non-negative slot by slot uncorrelated random variables with $E\{b(i)\}>E\{a(i)\}$ and set $u(i)=b(i)-a(i)$.
Then, equality \cite{kingman}
\begin{eqnarray}\label{eq_d1a}
  E\{u^2(i)\}-  2 E\{u(i)\} E\{q(i)\}= E\{(\max\{u(i)-q(i),0\})^2\}
\end{eqnarray}
and inequality \cite{daley} 
\begin{eqnarray}\label{eq_d1b}
    E\{(\max\{u(i)-q(i),0\})^2\}\geq \left(1-\frac{E\{a(i)\}}{E\{b(i)\}}\right) E\{b^2(i)\}
\end{eqnarray}
hold. Furthermore, combining (\ref{eq_d1a}) and (\ref{eq_d1b}), the following bound is obtained \cite{daley}
\begin{eqnarray}\label{eq_d1c}
    E\{q(i)\}\leq \frac{1}{2}\frac{E\{(b(i)-a(i))^2\}-(1-\xi)^2 E\{b^2(i)\}}{E\{b(i)\}-E\{a(i)\}},
\end{eqnarray}
where $\xi=E\{a(i)\}/E\{b(i)\}$.

By rewriting the queue size as
\begin{eqnarray}\label{eq_d1}
    Q(i)=\max\left\{Q(i-1)-d_i R(i)+(1-d_i) S(i),0\right\},
\end{eqnarray}
we observe that (\ref{eq_d1}) is in the form of (\ref{q-i}) if we let $q(i)=Q(i)$, $a(i)=(1-d_i)S(i)$, $b(i)=d_i R(i)$ and $\xi=E\{(1-d_i) S(i)\}/E\{d_i R(i)\}$. Thus, assuming that $s(i)$ and $r(i)$ are slot by slot uncorrelated, we can exploit (\ref{eq_d1c}) and upper bound  the average size of the queue as
\begin{eqnarray}\label{eq_d3}
    E\{Q(i)\}&\leq&  \frac{1}{2}\frac{E\{(1-d_i) S^2(i)\}+E\{d_i R^2(i)\}-2E\{(1-d_i)d_i S(i) R(i)\}-(1-\xi)^2 E\{d_i R^2(i)\} }{E\{d_i R(i)\} - E\{(1-d_i) S(i)\}}\nonumber\\
&=& \frac{1}{2} \frac{E\{(1-d_i) S^2(i)\}+\xi (2-\xi) E\{ d_i R^2(i) \} }{E\{d_i R(i)\} - E\{(1-d_i) S(i)\}}.
\end{eqnarray}
Since the average arrival rate is given by $A=E\{(1-d_i) S(i)\}$, we obtain (\ref{eq-f1}) from (\ref{eq_d3}) and Little's law (\ref{delay-main}). This completes the proof.
\end{appendix}
\bibliography{litdab}
\bibliographystyle{IEEETran}

\newpage
\begin{figure}
\includegraphics[width=6in ,height=2in]{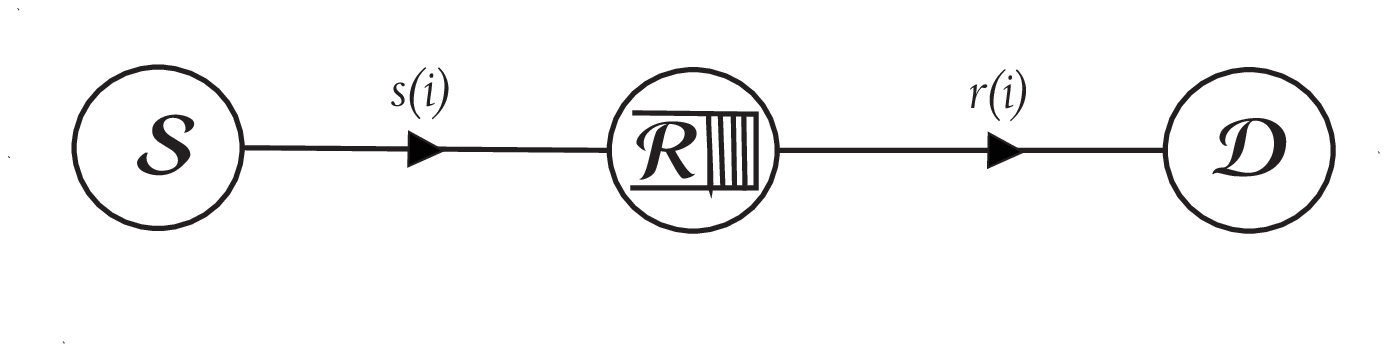}
\caption{System model comprising a source $\mathcal{S}$, a half-duplex relay equipped with a buffer $\mathcal{R}$, and a destination $\mathcal{D}$.  $s(i)$ and $r(i)$ are the instantaneous signal-to-noise ratios (SNRs) 
of the  $\mathcal{S}$-$\mathcal{R}$ and $\mathcal{R}$-$\mathcal{D}$ links in the $i$th time slot, respectively.}  \label{fig1}
\end{figure}
\begin{figure}
\includegraphics[width=7in ,height=5in]{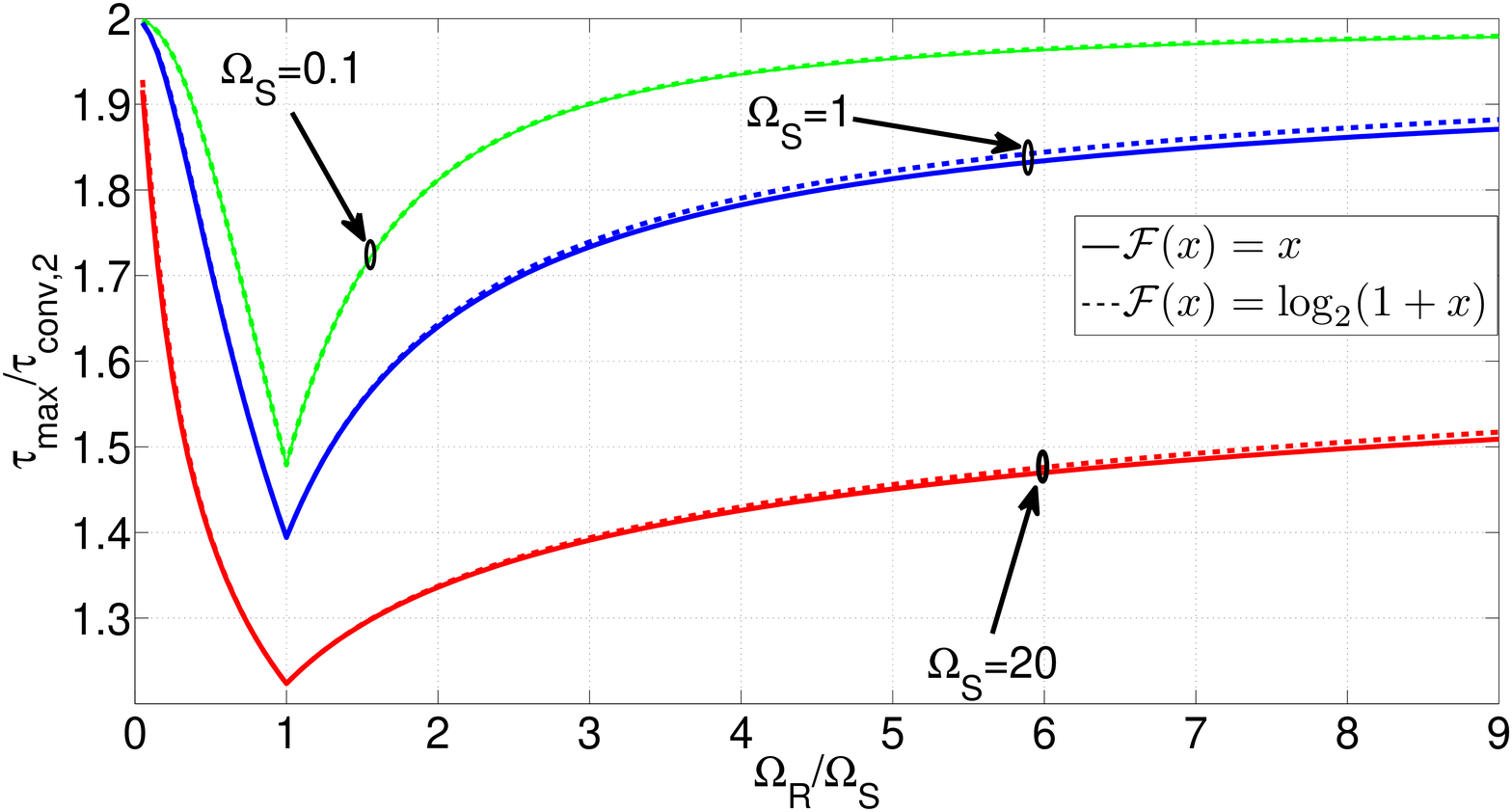}
\caption{Throughput ratio $\tau_{\max}/\tau_{\rm conv,2}$ vs.~$\Omega_R/\Omega_S$ for buffer-aided relaying with adaptive link selection and fixed transmit powers for source and relay.} \label{fig2}
\end{figure}
\begin{figure}
\includegraphics[width=7in ,height=5in]{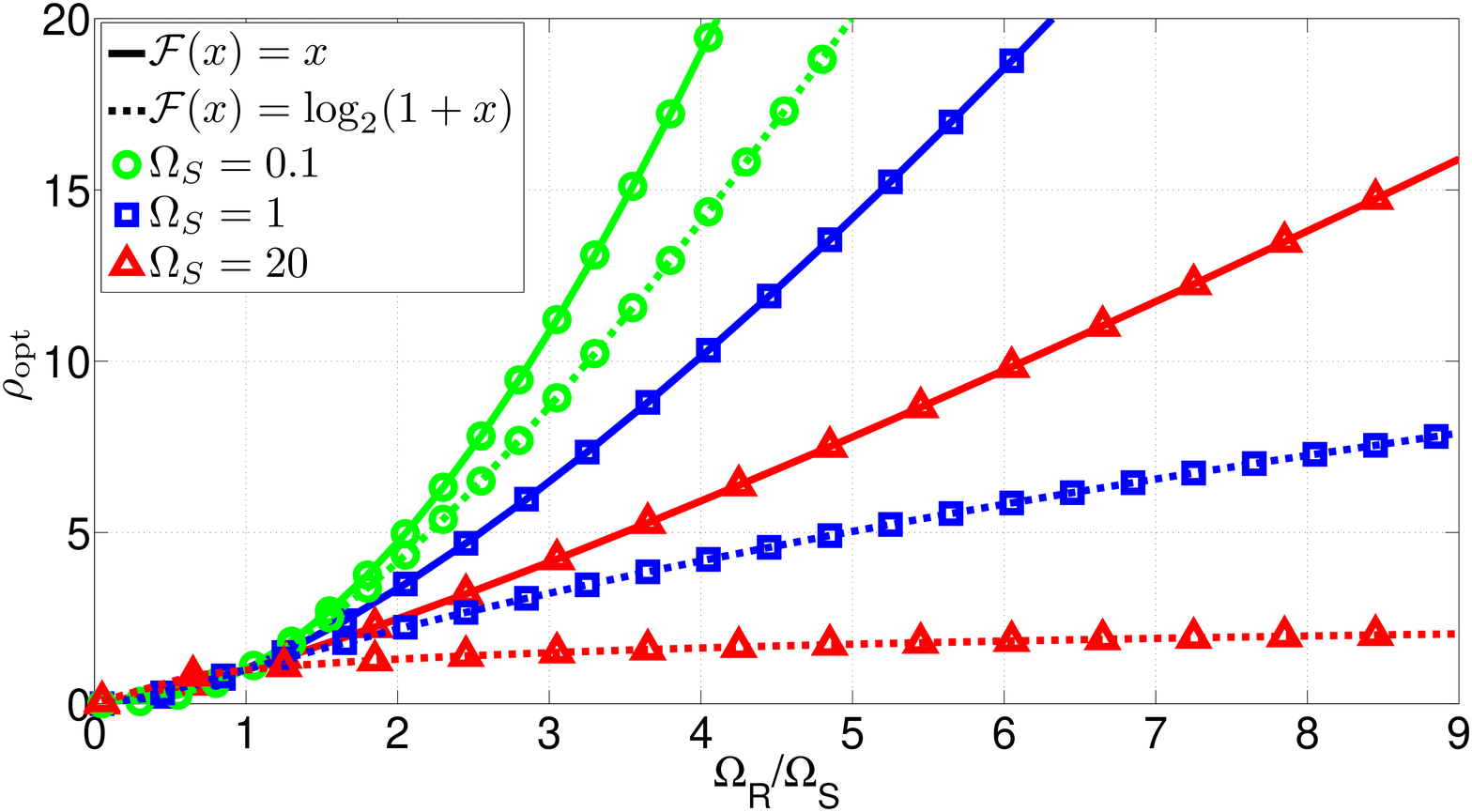}
\caption{Optimal decision threshold $\rho_{\rm opt}$ vs.~$\Omega_R/\Omega_S$ for buffer-aided relaying with adaptive link selection and fixed transmit powers for source and relay.} \label{fig3}
\end{figure}
\begin{figure}
\includegraphics[width=7in ,height=5in]{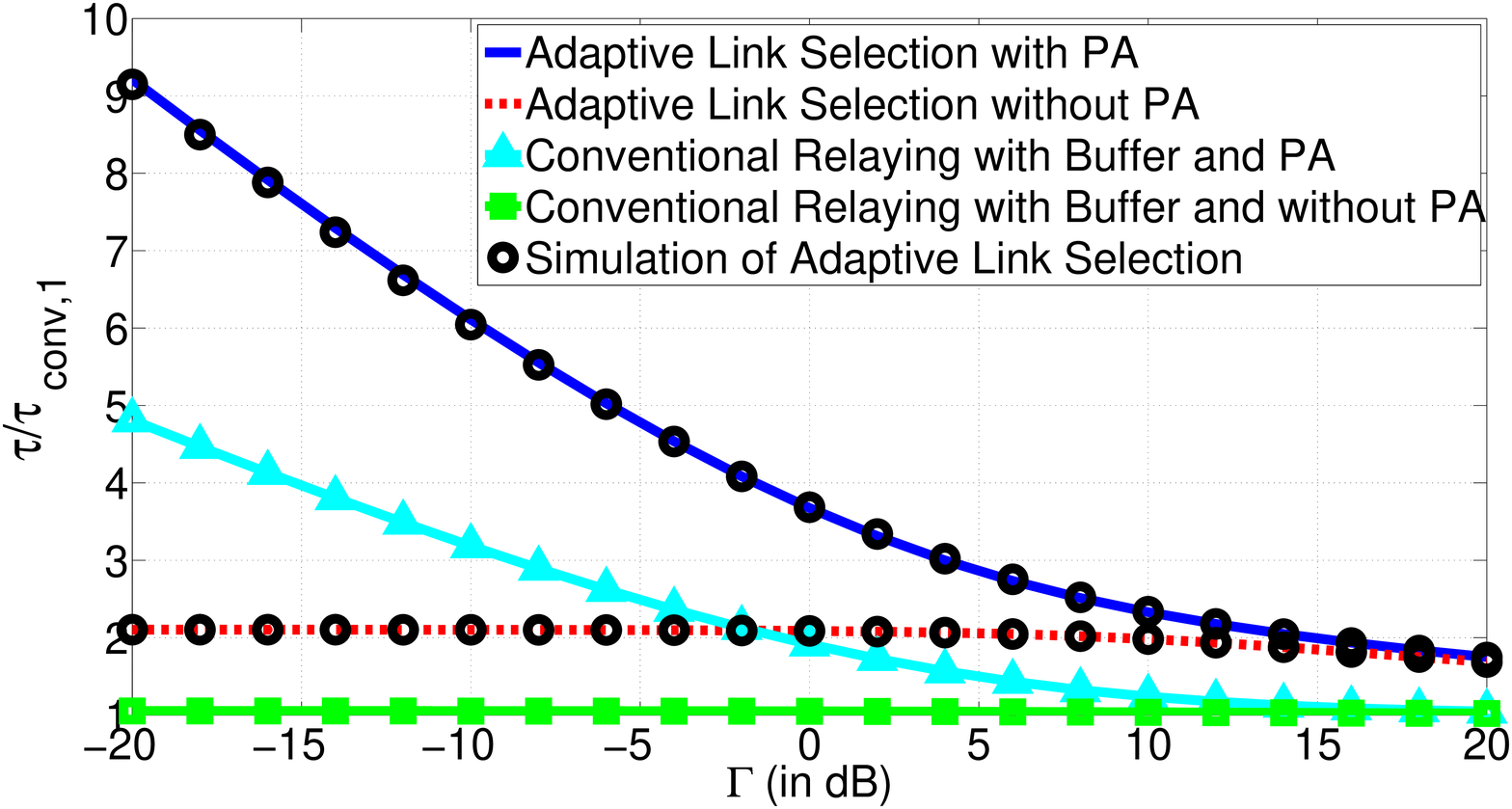}
\caption{Throughput normalized to $\tau_{\rm conv,1}$ vs.~$\Gamma$ for buffer-aided relaying with adaptive link selection and conventional relaying with buffer. The performance of both schemes with and without power allocation (PA) is shown. 
$\bar\Omega_S=0.1$ and $\bar\Omega_R=1.9$} \label{fig4b}
\end{figure}
\begin{figure}
\includegraphics[width=7in ,height=5in]{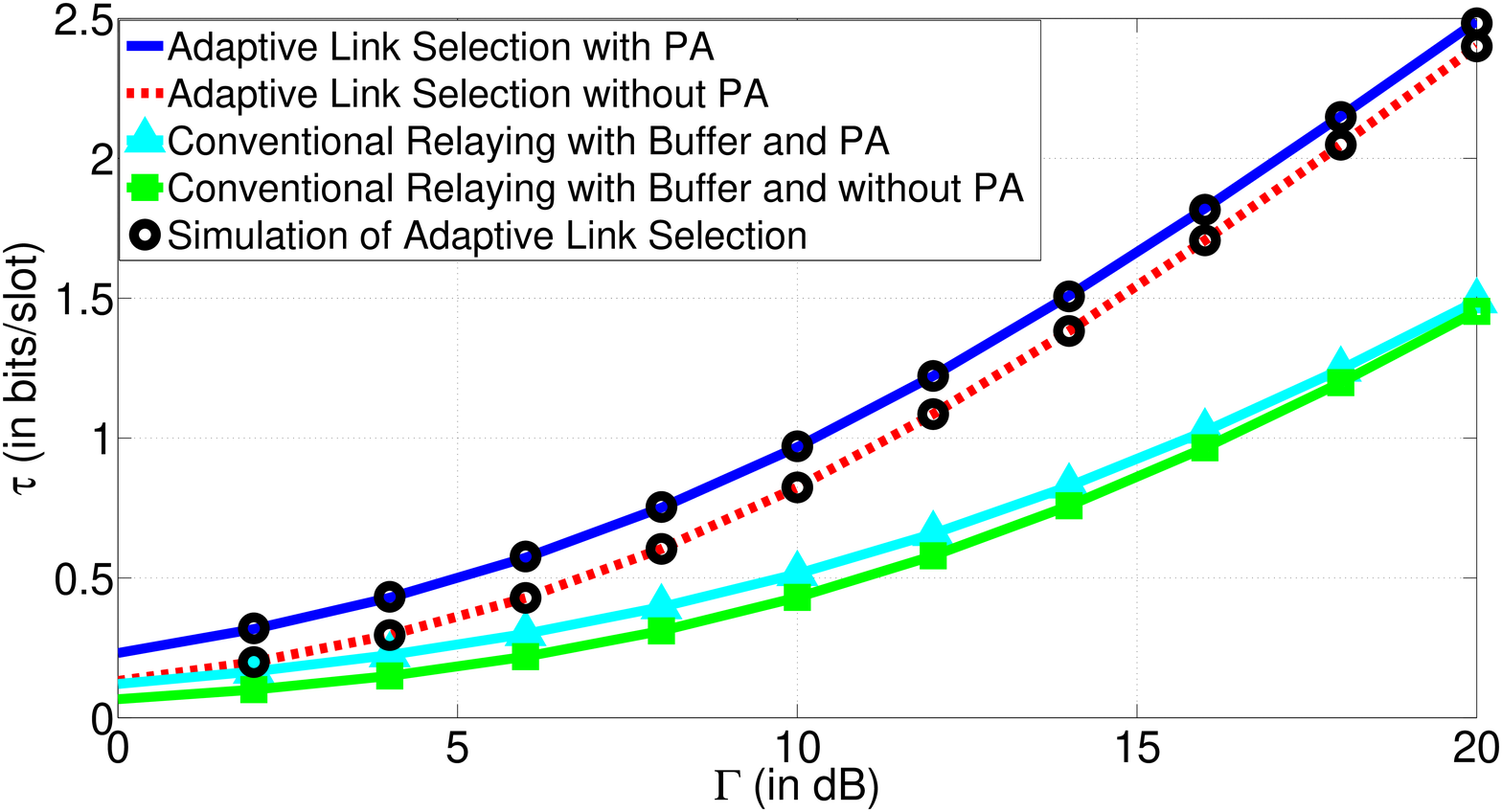}
\caption{Throughput vs.~$\Gamma$ for buffer-aided relaying with adaptive link selection and conventional relaying with buffer. The performance of both schemes with and without power allocation (PA) is shown. 
$\bar\Omega_S=0.1$ and $\bar\Omega_R=1.9$} \label{fig4c}
\end{figure}
\begin{figure}
\includegraphics[width=7in ,height=5in]{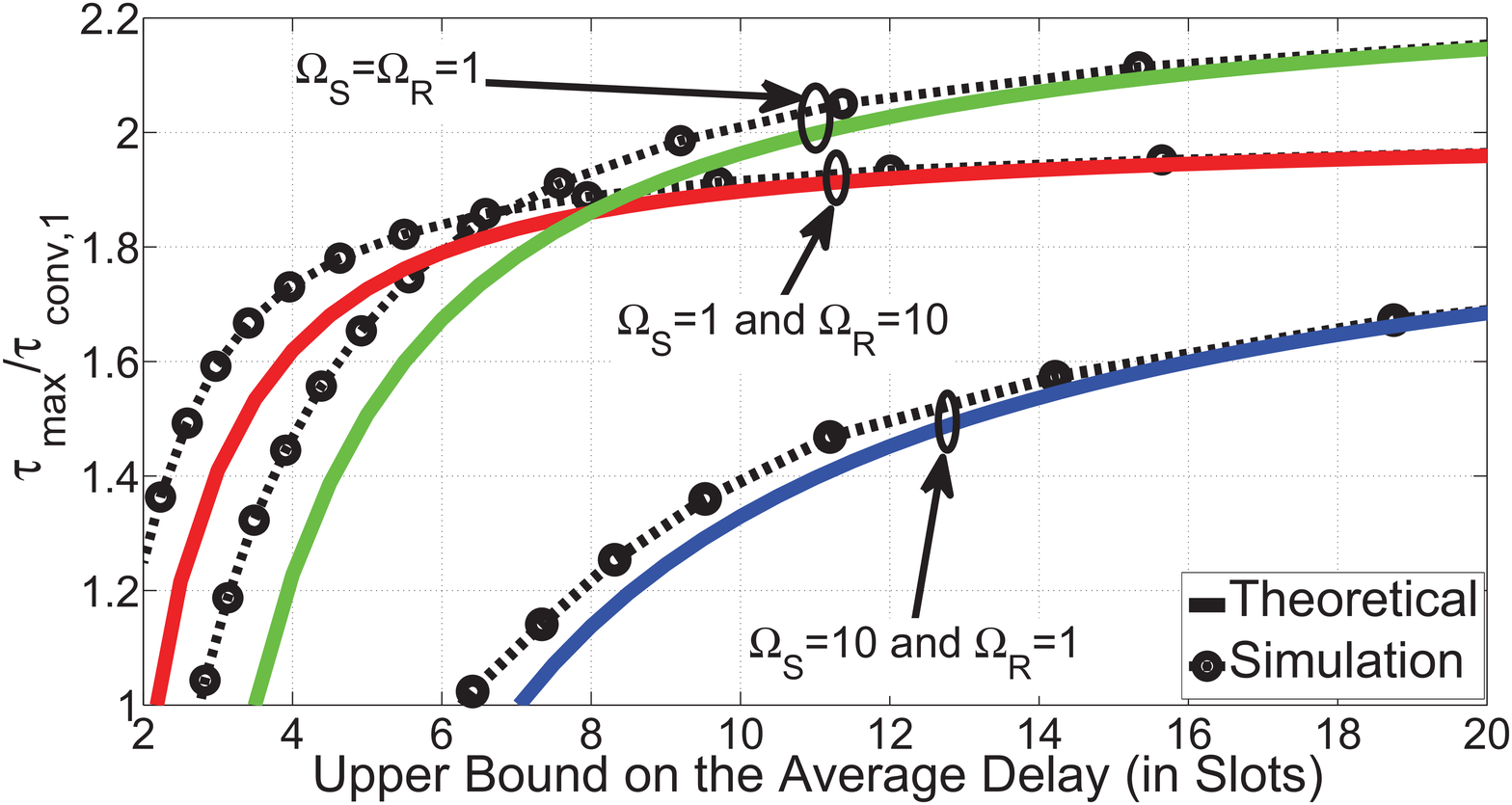}
\caption{Throughput ratio $\tau_{\max}/\tau_{\rm conv,1}$ vs.~upper bound on the average delay for buffer-aided relaying with adaptive link selection and starved buffer. $\mathcal{F}(x)=x$.} \label{fig5}
\end{figure}
\begin{figure}
\includegraphics[width=7in ,height=5in]{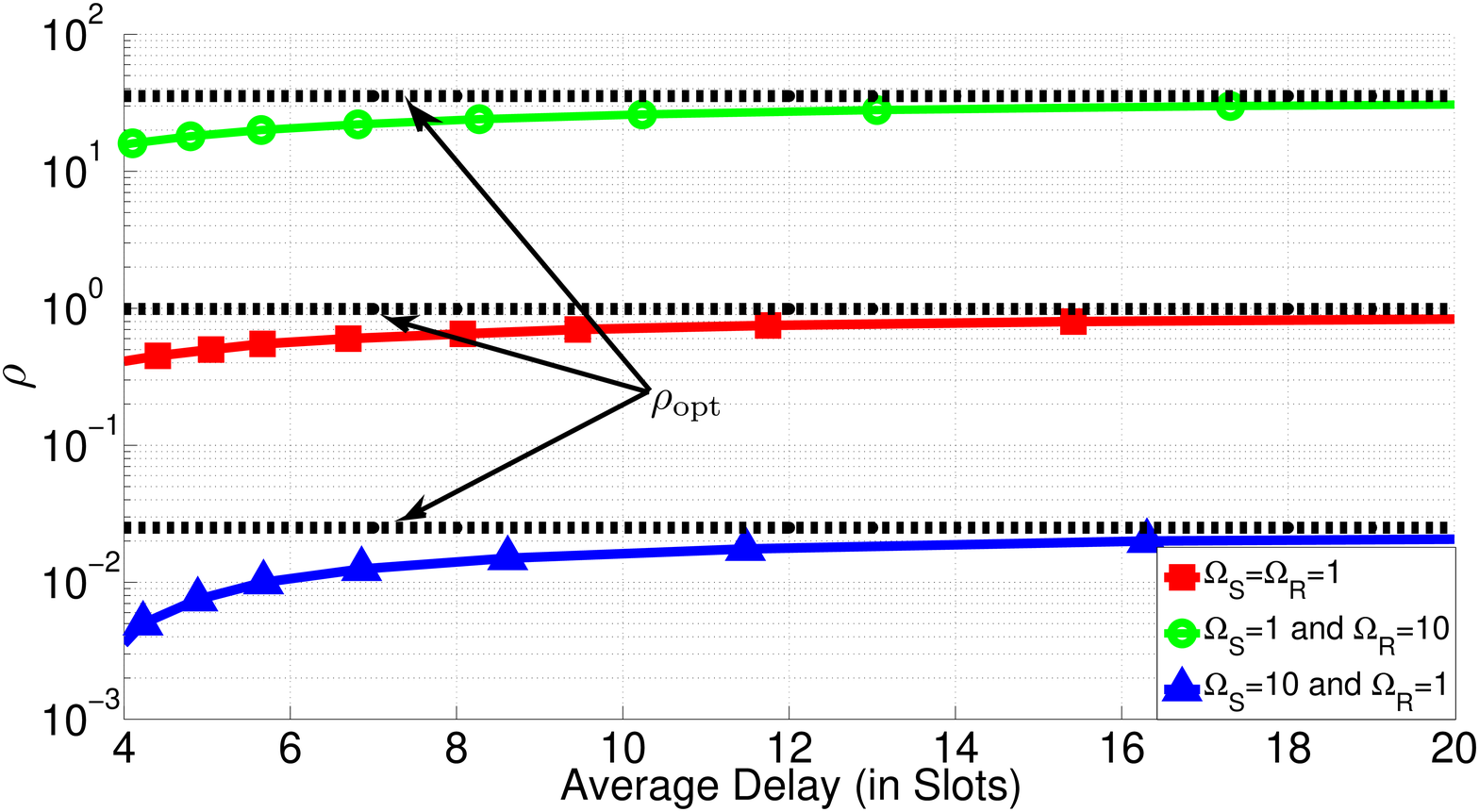}
\caption{Decision threshold $\rho$ vs.~upper bound on the average delay for buffer-aided relaying with adaptive link selection and starved buffer. $\mathcal{F}(x)=x$.} \label{fig6}
\end{figure}
\begin{figure}
\includegraphics[width=7in ,height=5in]{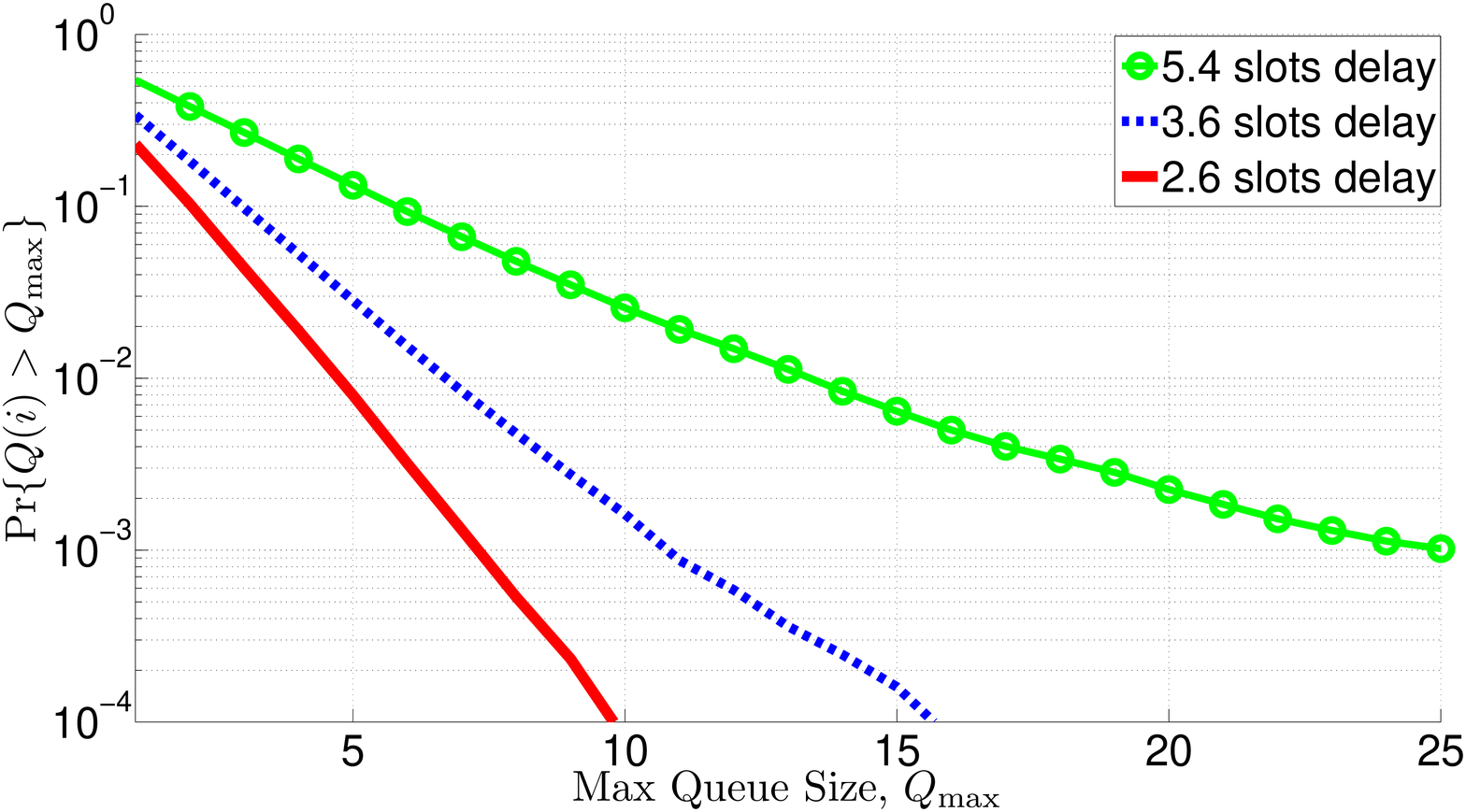}
\caption{Probability of a dropped bit vs.~queue size $Q_{\max}$ for buffer-aided relaying with adaptive link selection and starved buffer. $\Omega_S=\Omega_R=1$. $\mathcal{F}(x)=x$.} \label{fig7}
\end{figure}
\begin{figure}
\includegraphics[width=7in ,height=5in]{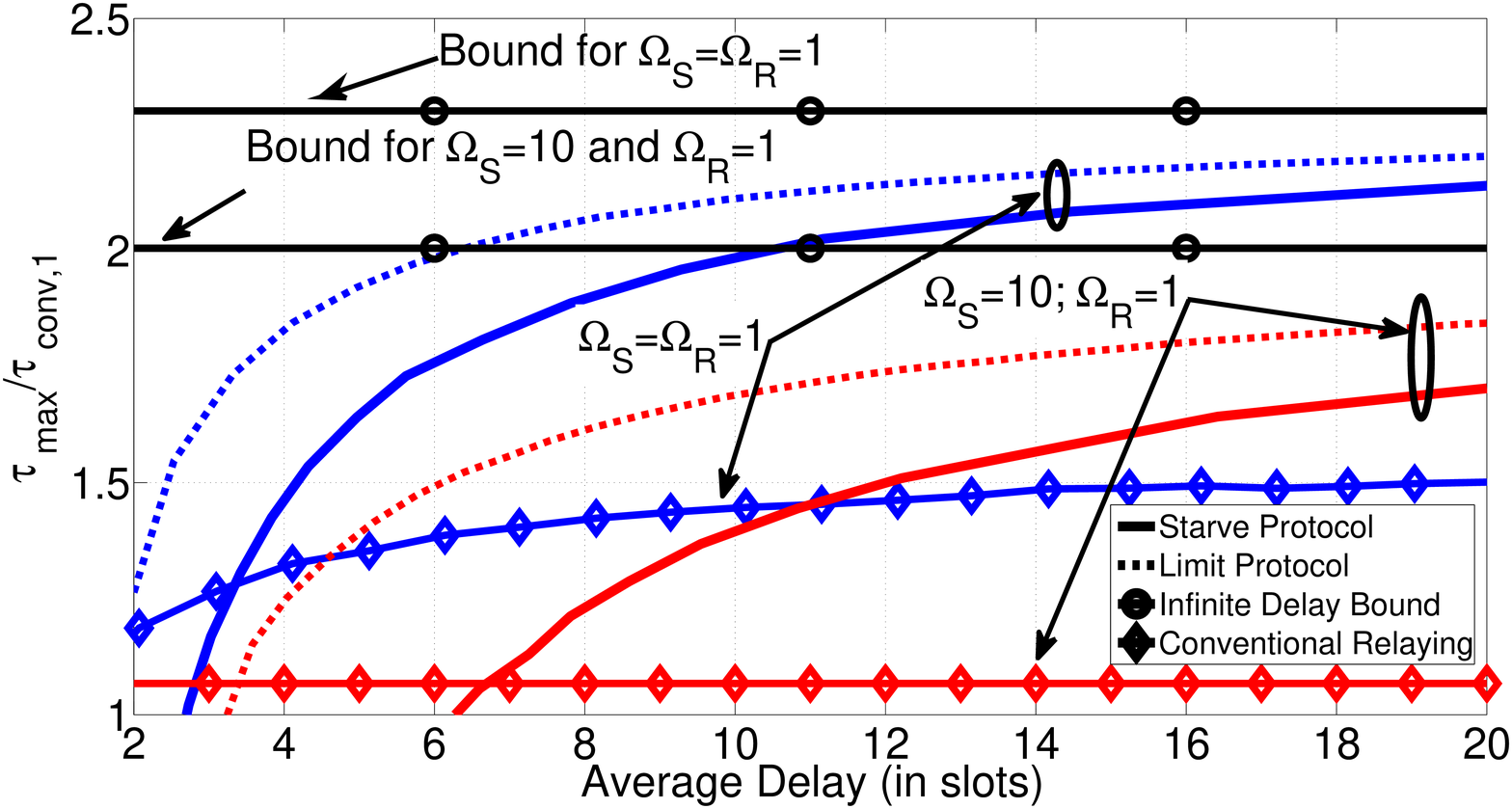}
\caption{Throughput ratio $\tau_{\rm max}/\tau_{\rm conv,1}$ vs.~average delay for the proposed buffer-aided relaying schemes with starved buffer and limited buffer and conventional relaying with buffer \cite{XFTP08}.} \label{fig8}
\end{figure}

\end{document}